\documentclass[lettersize, journal]{IEEEtran}
\usepackage{setspace}
\usepackage{cite}
\usepackage{pgfplots}

\pgfplotsset{compat=1.10}

\usepgfplotslibrary{fillbetween}

\usepackage{bm}
\usepackage{lipsum}
\usepackage{makeidx}
\usepackage{enumerate}
\usepackage{color}
\usepackage{cite}
\usepackage{amsmath,amsthm}   
\usepackage{amssymb}
\usepackage{nomencl}
\usepackage{multirow}
\usepackage{graphicx}
\usepackage{epstopdf}
\usepackage{threeparttable}
\usepackage{multicol}
\DeclareGraphicsExtensions{.pdf,.jpeg,.png,.jpg,.emf,.eps}
\usepackage{calc}
\usepackage{here}
\usepackage{url}

\usepackage{upref}
\usepackage{comment}
\usepackage{times}
\usepackage{dsfont}
\usepackage{epic,eepic}
\usepackage{rawfonts}
\usepackage[T1]{fontenc}
\usepackage{latexsym}
\usepackage{amsfonts}

\usepackage[font=small,labelfont=bf]{caption}
\usepackage[font=small,labelfont=bf]{subcaption}

\usepackage{xcolor}
\usepackage{tikz,pgfplots}
\usetikzlibrary{calc}
\usetikzlibrary{intersections}
\usetikzlibrary{arrows,shapes}
\usetikzlibrary{positioning}

\newtheorem{theorem}{Theorem}

\newtheorem{lemma}{Lemma}

\theoremstyle{definition}

\allowdisplaybreaks
\begin{document}

\title{A Framework for Fractional Matrix Programming Problems with Applications in FBL MU-MIMO
}
\author{Mohammad Soleymani, \emph{Member, IEEE},  
Eduard Jorswieck, \emph{Fellow, IEEE}, 
Robert Schober, \emph{Fellow, IEEE}, 
and
Lajos Hanzo, \emph{Life Fellow, IEEE}
 \\ \thanks{ 
Mohammad Soleymani is with the Signal and System Theory Group, University of Paderborn, 33098 Paderborn, Germany (email: \protect\url{mohammad.soleymani@uni-paderborn.de}).

Eduard Jorswieck is with the Institute for Communications Technology, Technische Universit\"at Braunschweig, 38106 Braunschweig, Germany (email: \protect\url{jorswieck@ifn.ing.tu-bs.de}).

Robert Schober is with the Institute for Digital Communications, Friedrich Alexander University of Erlangen-Nuremberg, 91058 Erlangen, Germany (email: \protect\url{robert.schober@fau.de}).

Lajos Hanzo is with the Department of Electronics and Computer Science, University of Southampton,, SO17 1BJ Southampton, United Kingdom (email:\protect\url{lh@ecs.soton.ac.uk}).

The work of Eduard A. Jorswieck was supported in part by the Federal Ministry for Research, Technology and Space (BMFTR) in Germany in the programme of “Souver\"an. Digital. Vernetzt.” Joint project 6G-RIC, project identification number: 16KISK031, and by European Union's (EU's) Horizon Europe project 6G-SENSES under Grant 101139282. Robert Schober’s work was supported by the Federal Ministry for Research, Technology and Space (BMFTR) in Germany in the program of ``Souver\"an. Digital. Vernetzt.'' joint project 6G-RIC (Project-ID 16KISK023) and the Deutsche Forschungsgemeinschaft (DFG, German Research Foundation) under projects SFB 1483 (Project-ID 442419336, EmpkinS). Lajos Hanzo would like to acknowledge the financial support of the Engineering and Physical Sciences Research Council (EPSRC) projects under grant EP/X01228X/1, EP/Y026721/1, EP/W032635/1 and EP/X04047X/1.
}}

\maketitle
\begin{abstract}
An efficient framework is conceived for fractional matrix programming (FMP) optimization problems (OPs) namely for minimization and maximization. In each generic OP, either the objective or the constraints are functions of multiple arbitrary continuous-domain fractional functions (FFs). This ensures the framework's versatility, enabling it to solve a broader range of OPs than classical FMP solvers, like Dinkelbach-based algorithms. Specifically, the generalized Dinkelbach algorithm can only solve multiple-ratio FMP problems. By contrast, our framework solves OPs associated with a sum or product of multiple FFs as the objective or constraint functions. Additionally, our framework provides a single-loop solution, while most FMP solvers require twin-loop algorithms. 

Many popular performance metrics of wireless communications are FFs. For instance, latency has a fractional structure, and minimizing the sum delay leads to an FMP problem. Moreover, the mean square error (MSE) and energy efficiency (EE) metrics have fractional structures. Thus, optimizing EE-related metrics such as the sum or geometric mean of EEs and enhancing the metrics related to spectral-versus-energy-efficiency tradeoff yield FMP problems. Furthermore, both the signal-to-interference-plus-noise ratio and the channel dispersion are FFs. In this paper, we also develop resource allocation schemes for multi-user multiple-input multiple-output (MU-MIMO) systems, using finite block length (FBL) coding, demonstrating attractive practical applications of FMP by optimizing the aforementioned metrics.

\end{abstract}
\begin{IEEEkeywords}
Finite block length coding, fractional matrix programming, latency minimization, mean square error, multi-user MIMO systems,  reconfigurable intelligent surface, spectral-energy efficiency tradeoff.
\end{IEEEkeywords}
\section{Introduction}
Fractional matrix programming (FMP) problems encompass a wide range of practical optimization problems (OPs) in wireless communications, including latency minimization, energy efficiency (EE) maximization, and mean square error (MSE) minimization. These OPs are critical, since 6G is expected to substantially improve latency, EE, and reliability compared to current wireless systems \cite{wang2023road, gong2022holographic, CIT-129}. FMP OPs are typically non-convex and very challenging to solve, even when the optimization variables are scalar or in vector format. Modern wireless communication systems mainly utilize multiple-antenna devices for enhancing both the spectral efficiency (SE) and EE. To effectively allocate resources in multiple-antenna systems, typically matrices have to be optimized, which further complicates the corresponding OPs. This issue can be even more difficult to address, when low-latency requirements should be fulfilled. In this case, finite block length (FBL) coding may have to be employed, which makes the classical Shannon rate an inaccurate metric \cite{vaezi2022cellular, polyanskiy2010channel}. In this paper, we propose a general framework for solving FMP problems, focusing on applications in multi-user (MU) multiple-input multiple-output (MIMO) systems with FBL coding.

There are several bespoke solvers for FMP problems, including the family of Dinkelbach-based algorithms \cite{dinkelbach1967nonlinear, crouzeix1991algorithms}, the Charnes-Cooper transform \cite{cooper1962programming, schaible1974parameter}, and the quadratic transform proposed by Shen and Yu in \cite{shen2018fractional}. The Dinkelbach algorithm and the Charnes-Cooper transform can solve OPs, having a single-ratio fractional function (FF) as the objective function (OF). The Dinkelabch algorithm can obtain the globally optimal solution of single-ratio fractional maximization problems when the numerator and denominator of the OF are, respectively, convex and concave. However, when the numerator is non-concave and/or the denominator is non-convex, the Dinkelabch algorithm converges to a stationary point (SP) of the OP. In this case, one usually has to employ a twin-loop approach, when using the Dinkelbach algorithm. Moreover, this algorithm cannot solve more complicated FMP OPs in the family of multiple-ratio FMPs.

To solve max-min (or min-max) multiple-ratio FMP OPs, one can employ the generalized Dinkelbach algorithm (GDA), which converges to an SP of the original OP \cite{crouzeix1991algorithms}. Indeed, unlike the Dinkelbach algorithm, the GDA does not converge to the globally optimal solution. Moreover, the GDA is a twin-loop algorithm when optimizing, for instance, the minimum of EE among multiple users. More precisely, if the numerator (and/or denominator) of even one of the FFs is non-concave (and/or non-convex), another numerical solver such as majorization minimization (MM) has to be employed in addition to the GDA to solve the FMP OP. This yields a two-loop algorithm, which might be challenging to implement and has poor convergence.   

Even though Dinkelbach-based algorithms lead themselves to the solution of both single- and multiple-ratio FMP OPs, they cannot solve more complicated FMP OPs. More specifically, when the OF and/or the constraint functions (CFs) of an FMP OP are consisted by a sum or a product of multiple FFs, Dinkelbach-based algorithms cannot solve the OP. Unfortunately, many practical FMP OPs fall into this category. This type of FMP OPs includes, for instance, the sum delay minimization \cite{peng2023non, buzzi2016survey}, the product or geometric mean of EE maximization \cite{buzzi2016survey, zappone2015energy}, and the SEE tradeoff enhancement \cite{zappone2023rate, aydin2017energy, you2020energy, you2020spectral}. Additionally, Dinkelbach-based algorithms cannot solve FMP OPs in which either the OF or the CFs are in matrix format, such as the MSE metric of MU-MIMO systems, as depicted in Table \ref{tab:com-gda}.

\begin{table}
\caption{Comparison of our framework and Dinkelbach-based algorithms, based on the problems they can solve.}
\scriptsize
    \centering
    \begin{tabular}{|l||c|c|c|c|c|c|c|c|c|c|c|}
    \hline
         &  Our framework & Dinlekbach algorithm & GDA 
        \\       
        \hline
         Single-ratio FFs &$\surd$ &$\surd$&-
         \\       
        \hline
         Multi-ratio FFs &$\surd$ &-&$\surd$
         \\       
        \hline
         Sum of FFs &$\surd$ &-&-
         \\       
        \hline
        Product of FFs&$\surd$ &-&-
        \\       
        \hline
        MSE Optimization&$\surd$ &-&-
        \\       
        \hline
        Single-loop Implementation &$\surd$ &-&-\\
        \hline
    \end{tabular}  
    \label{tab:com-gda}
    \normalsize
\end{table}

 The authors of \cite{shen2018fractional} proposed a quadratic transform for solving continuous fractional programming (FP) maximization problems. The algorithm proposed in \cite{shen2018fractional} solves a wider range of OPs than classical FP solvers, such as Dinkelbach-based algorithms. The authors of \cite{shen2018fractional} developed both power allocation and beamforming vector designs for various wireless systems, such as point-to-point communications and broadcast channels (BCs). The quadratic transform proposed in \cite{shen2018fractional} introduces a new optimization variable for each FF to convert the FF to a non-fractional format. Then, alternating optimization (AO) is employed to separate the optimization of the powers (or beamforming vectors) and the optimization of the variables introduced by the transformation. 

 By contrast, we propose a novel unified framework for solving practical FMP OPs arising in wireless communication system design, with an emphasis on low-latency applications in MU-MIMO systems. To this end, we solve a generic minimization and a generic maximization problem. In these OPs, the optimization variables are matrices, in contrast to the framework in \cite{shen2018fractional}, which focuses on optimizing scalar and vector variables. Additionally, the authors of \cite{shen2018fractional} developed a solver suitable only for maximization problems, while we propose frameworks for both maximization and minimization FMP OPs. Moreover, the OF or CFs of these FMP OPs can be arbitrary continuous functions of FFs. Furthermore, the numerator and denominator of each FF can be any arbitrary continuous function.

In addition to developing a general framework for FMP, we propose resource allocation schemes for MU-MIMO systems relying on FBL coding. There is a parcity of literature on resource allocation for MU-MIMO systems with FBL coding and multi-stream data transmission per user \cite{soleymani2024optimization, soleymani2024rate}. The paper in \cite{soleymani2024optimization} obtained a closed-form expression for the FBL rate of MU-MIMO systems in a matrix format and optimized the SE and EE of MU-MIMO ultra-reliable low-latency communication (URLLC) systems. Additionally, in \cite{soleymani2024rate}, it was shown that rate splitting multiple access (RSMA) enhances the maximum of the minimum rate and the minimum EE among multiple users in a multi-cell MU-MIMO URLLC BC. Even though the EE metrics considered in  \cite{soleymani2024optimization, soleymani2024rate} are in fractional form, the OF and CFs of the OPs studied there do not contain summations or products of multiple FFs, which are the main focus of this treatise. 

 In this paper, we first solve the considered generic OPs and then provide several practical examples for each generic FMP OP, specifically emphasizing  performance metrics that are not studied in \cite{soleymani2024optimization, soleymani2024rate, shen2018fractional}, as illustrated in Table \ref{tab:sec-i}. These metrics include the sum delay, weighted sum EE, geometric mean of delays, geometric mean of EEs, and SEE tradeoff metrics. Additionally, we consider the MSE of MU-MIMO systems, when there are multiple data streams per user. These performance metrics are essential for enhancing the SE, EE, reliability, and latency. Furthermore, we study a broader range of MU-MIMO systems compared to \cite{ shen2018fractional}.

\begin{table}
\centering
\scriptsize
\caption{Overview of most closely related works.}\label{tab:sec-i}
\begin{tabular}{|l||c|c|c|c|c|c|c|c|c|c|c|c|c|c|}
 \hline	
 &This paper&\cite{shen2018fractional}&\cite{soleymani2024optimization}&\cite{soleymani2024rate}\\
 \hline
MIMO systems&
$\surd$&$\surd$&$\surd$&$\surd$
\\
\hline
Multiple-stream data transmission per user&
$\surd$&-&$\surd$&$\surd$
\\
\hline
FBL coding&
$\surd$&-&$\surd$&$\surd$
\\
\hline
RIS&
$\surd$&-&$\surd$&$\surd$
\\
\hline
\textbf{SEE tradeoff}&$\surd$&-&-&-
\\
\hline
\textbf{Sum/geometric mean delays}&$\surd$&-&-&-
\\
\hline
\textbf{Sum/geometric mean EEs}&$\surd$&-&-&-
\\
\hline
\textbf{MSE optimization for MU-MIMO systems}&$\surd$&-&-&-
\\
\hline
		\end{tabular} 
\normalsize
\end{table}

Finally, the main contributions of this paper can be summarized as follows:
\begin{itemize}
   \item We propose a framework for FMP problems, solving a pair of generic OPs: a minimization and a maximization problem. In these generic OPs, the OF and CFs are functions of arbitrary continuous-domain non-negative FFs. These OPs encompass a wide range of practical performance metrics and problems in wireless communications, encompassing latency, reliability, EE, and SEE tradeoff metrics. Our FMP solver is an iterative algorithm converging to an SP of the original OP.  

   \item The main focus of our framework is on matrix optimizations, capable of optimizing any arbitrary continuous-domain complex-valued matrices. Hence, one can leverage our framework for optimizing beamforming matrices at transceivers, including precoders and decoders, or RIS coefficients in RIS-assisted systems. 

   \item We propose a single-loop algorithm for solving complex FMP OPs. Most FMP solvers require a twin-loop implementation, when the FFs are neither concave nor convex. However, our framework requires only a single loop, substantially reducing the implementation and computational complexities.

    \item We also propose resource allocation schemes for MU-MIMO systems relying on FBL coding. More specifically, we consider a BC, operating with and without an RIS and solve various practical OPs, including the maximization of the SEE tradeoff, the weighted sum EE, the geometric mean of EE, and the minimization of sum delay, the geometric mean of delays, maximum MSE, and weighted sum MSE. These FMP OPs are not solvable by classical FMP OPs, since the OF or CFs of these problems are either summations or products of FFs. However, our framework converges to an SP of these FMP OPs.

\end{itemize}

The paper is organized as follows. Section \ref{secii-new} presents the main results of this treatise. More specifically, Section \ref{secii-new} formulates the pair of generic FMP OPs we tackle in this paper  and proposes solutions for each. Section \ref{sec=iii} provides four practical examples of FMP minimization problems that our framework can solve, but Dinkelbach-based algorithms cannot. Section \ref{sec=iii-max} solves four commonly used FMP maximization problems that are not solvable by Dinkelbach-based algorithms. Section \ref{sec=iv} shows how the framework can be extended to RIS-aided systems. Section \ref{sec-dink} extensively compares the performance and computational complexity of our FMP solver to Dinkelbach-based algorithms. Section \ref{sec-new-vii} explores several potential directions for extending this line of research. Finally, Section \ref{sec-conclusion} concludes the paper.

\begin{table}
\centering
\scriptsize
\caption{List of most frequently used abbreviations.}\label{tab:iii}
\begin{tabular}{|l|l||l|l|c|c|c|c|c|c|c|c|c|c|c|}
 \hline	
GDA & Generalized Dinkelbach algorithm&HWI & Hardware impairment
\\
\hline	
FMP & Fractional matrix programming&OP & Optimization problem
\\\hline	
MM & Majorization minimization&CF & Constraint function
\\\hline	
CSI & Channel state information&FF & Fractional function
\\\hline	
SEE & Spectral-energy efficiency&OF & Objective function
\\\hline	
AO & Alternating optimization&GM & Geometric mean
\\\hline	
FP & Fractional programming&SP & Stationary point
\\\hline
		\end{tabular} 
\normalsize
\end{table}

\section{Problem Statement and Proposed Solutions}\label{secii-new} 
In this section, we formulate the OPs addressed and present our solutions.  {To solve these general FMP OPs, we propose iterative algorithms. Inspired by the MM framework\footnote{Please refer to Appendix \ref{app-mm} for preliminaries on the MM framework.}, each iteration of our algorithms consists of two steps. In the first step, we reformulate the FMP problem into a simpler preferably non-FMP problem, referred to as a surrogate problem. In the second step, we solve the surrogate problem to obtain an updated point. This procedure is iterated until fulfilling the convergence criteria.} Our minimization and maximization problems are discussed in separate subsections below.

\subsection{Minimization Problems}\label{sec-ii-min}
In this subsection, we aim for solving FMP minimization problems obeying the following format 
    \begin{align}\label{(1-min)}
        \min_{\{{\bf X}   \}\in \mathcal{X}  } & \sum_{m=1}^{M_0}h_{m0}(\{{\bf X}   \})
        &%\\
        \text{s.t.}&\sum_{m=1}^{M_i}h_{mi}(\{{\bf X}   \})\leq \eta_i, \, \forall i,
    \end{align}
where $\{{\bf X}   \}=\{{\bf X}_1,{\bf X}_2,\cdots,{\bf X}_N\}$ is the set of optimization variables, and $\eta_i>0$, $\forall i$, is a known constant. Here, ${\bf X}_n$, $\forall n$, is an arbitrary complex matrix, $\mathcal{X}$ is the feasibility set of the variables, which is assumed to be a convex set, and $h_{m,i}(\{{\bf X}   \})$ is an FF, i.e., we have
\begin{equation}
    h_{mi}(\{{\bf X}   \})=\frac{f_{mi}(\{{\bf X}   \})}{g_{mi}(\{{\bf X}   \})},
\end{equation}
where $f_{mi}(\{{\bf X}   \})\geq 0$ and $g_{mi}(\{{\bf X}   \})>0$, $\forall m,i,$  are arbitrary continuous functions of $\{{\bf X}   \}$.  The OP in \eqref{(1-min)} is a non-convex FMP minimization problem, which cannot be solved by classical FMP solvers such as Dinkelbach-based algorithms. In the following theorem, we convert \eqref{(1-min)} into a simpler OP, allowing the reuse of optimization techniques such as MM. 
\begin{theorem}\label{th1-min}
    A stationary point of \eqref{(1-min)} can be obtained by iteratively solving 
    \begin{subequations}\label{(2-min)}
    \begin{align}
        \min_{\{{\bf X}   \}\in \mathcal{X}, {\bf t}, {\bf u}  } & \sum_{m=1}^{M_0}\frac{u_{m0}^{2}}{t_{m0}}
        \\
        \label{eq-2b-min}\text{\em s.t.}&\sum_{m=1}^{M_i}\frac{u_{mi}^{2}}{t_{mi}}\leq \eta_i, \forall i,
        \\
\label{eq-2d-min}        & g_{mi}(\{{\bf X}\})-t_{mi}\geq 0, \, \forall m,i,
        \\
        \label{eq-2c-min}        & 2\sqrt{f_{mi}^{(z)}}u_{mi} - f_{mi}^{(z)}\geq f_{mi}(\{{\bf X}\}),\,\, \forall m,i,
    \\
    \label{eq-2e-min}&
        t_{mi}> 0,\,\,  u_{mi}\geq 0, \, \forall m,i,
    \end{align}
\end{subequations}
where $ {\bf t}=\{t_{mi} \}_{\forall mi}$ and $ {\bf u}=\{u_{mi} \}_{\forall mi}$ are the sets of auxiliary optimization variables,  $z$ is the iteration index, $f_{mi}^{(z)}=f_{mi}(\{{\bf X}^{(z)}\})$, and 
 $\{{\bf X}^{(z)}\}=\{{\bf X}_1^{(z)},{\bf X}_2^{(z)},\cdots,{\bf X}_N^{(z)}\}$.
\end{theorem}
\begin{proof}
    Please refer to Appendix \ref{app-1-min}.
\end{proof}
\subsubsection{When $f_{mi}$ is convex and $g_{mi}$ is concave $\forall m,i$} 
In this case, \eqref{(2)} is a convex OP, solvable by numerical tools such as CVX \cite{grant2014cvx}. Moreover, the framework converges to an SP of \eqref{(1-min)} since the approximation in \eqref{eq-2c-min} fulfills the three conditions in Appendix \ref{app-mm}, as shown in the proof of Theorem \ref{th1-min}.

\subsubsection{General case}
Unfortunately, in many practical FMP problems of wireless communications, the functions $f_{mi}(\{{\bf X}   \})$ and $g_{mi}(\{{\bf X}   \})$ are neither convex nor concave in $\{{\bf X}\}$. In this case, an SP of \eqref{(1-min)} can be found using the following Lemma.
\begin{lemma}\label{cor-min}
    When $f_{mi}(\{{\bf X}   \})$ and $g_{mi}(\{{\bf X}   \})$ are neither convex nor concave, a stationary point of \eqref{(1-min)} can be found by iteratively solving 
    \begin{subequations}\label{(2-min-cor)}
    \begin{align}
        \min_{\{{\bf X}   \}\in \mathcal{X}, {\bf t}, {\bf u}  } & \sum_{m=1}^{M_0}\frac{u_{m0}^{2}}{t_{m0}}
        \\
        \text{\em s.t.}\,\,\,\, & \tilde{g}_{mi}(\{{\bf X}\})-t_{mi}\geq 0, \, \forall m,i,
         \\
         & 2\sqrt{f_{mi}^{(z)}}u_{mi} - f_{mi}^{(z)}\geq \tilde{f}_{mi}(\{{\bf X}\}),\,\, \forall m,i,,
    \\         
         &\eqref{eq-2b-min}, \eqref{eq-2e-min}
    \end{align}
\end{subequations}
where $ {\bf t}$, $ {\bf u}$,  $z$, and 
 $\{{\bf X}^{(z)}\}=\{{\bf X}_1^{(z)},{\bf X}_2^{(z)},\cdots,{\bf X}_N^{(z)}\}$ are defined as in Theorem \ref{th1-min}. Additionally, $\tilde{f}_{mi}(\{{\bf X}   \})$ is a convex \textit{upper bound} for $f_{mi}(\{{\bf X}   \})$, and $\tilde{g}_{mi}(\{{\bf X}   \})$ is a concave \textit{lower bound} for $g_{mi}(\{{\bf X}   \})$, satisfying the following conditions:
\begin{itemize}
    \item $\tilde{f}_{mi}(\{{\bf X}^{(z)}   \})=f_{mi}(\{{\bf X}^{(z)}   \})$ and  $\tilde{g}_{mi}(\{{\bf X}^{(z)}   \})=g_{mi}(\{{\bf X}^{(z)}   \})$, 
    \item $\tilde{f}_{mi}(\{{\bf X}   \})>f_{mi}(\{{\bf X}   \})$ and  $\tilde{g}_{mi}(\{{\bf X}   \})<g_{mi}(\{{\bf X}   \})$, 
    \item $\left. \frac{\partial \tilde{f}_{mi}(\{{\bf X}   \})}{\partial \{{\bf X}   \}}\right|_{\{{\bf X}   \}=\{{\bf X}^{(z)}\}}
    =\left. \frac{\partial {f}_{mi}(\{{\bf X}   \})}{\partial \{{\bf X}   \}}\right|_{\{{\bf X}   \}=\{{\bf X}^{(z)}\}}$ and  $\left. \frac{\partial \tilde{g}_{mi}(\{{\bf X}   \})}{\partial \{{\bf X}   \}}\right|_{\{{\bf X}   \}=\{{\bf X}^{(z)}\}}
    =\left. \frac{\partial {g}_{mi}(\{{\bf X}   \})}{\partial \{{\bf X}   \}}\right|_{\{{\bf X}   \}=\{{\bf X}^{(z)}\}}$, 
\end{itemize}
for all $m,i$ and $\{{\bf X}   \}\in \mathcal{X}$.
\end{lemma}
\begin{proof}
    Please refer to Appendix \ref{app-prof-lem1}. 
\end{proof}
 {Inspired by MM, our framework utilizes an iterative approach, consisting of two steps in each iteration: majorization and minimization. Indeed, we majorize $h_{mi}(\{{\bf X}\})$ to reformulate \eqref{(1-min)} as a simpler surrogate OP, as shown in Fig. \ref{Fig-mm}a. When $f_{mi}$ is non-convex and/or $g_{mi}$ is  non-concave, we employ the following surrogate function 
\begin{equation}\label{eq:sur-func-h}
    \tilde{h}_{mi}(\{{\bf X}   \})=\frac{\tilde{f}_{mi}(\{{\bf X}   \})}{\tilde{g}_{mi}(\{{\bf X}   \})}\geq h_{mi}(\{{\bf X}   \}),\,\,\, \forall m,i,
\end{equation}
where $\tilde{f}_{mi}(\{{\bf X}   \})$ and $\tilde{g}_{mi}(\{{\bf X}   \})$ fulfill the conditions formulated in Lemma \ref{cor-min}.}
Algorithm I summarizes our solution for \eqref{(1-min)} in the general case.
\doublespacing 
\begin{table}[htb]
\small
\begin{tabular}{l}
\hline 
 \textbf{Algorithm I}: Our framework to solve \eqref{(1-min)} in the general case.  \\
\hline
\hspace{0.2cm}{\textbf{Initialization}}\\
\hspace{0.2cm}Set $\delta$,  $z=1$,  $\{\mathbf{X}\}=\{\mathbf{X}^{(0)}\}$\\
\hline %\hline
\hspace{0.2cm}
\textbf{While} $
\frac{
|\sum_{m=1}^{M_0}h_{m0}(\{{\bf X}^{(z)}   \})
-
\sum_{m=1}^{M_0}h_{m0}(\{{\bf X}^{(z-1)}   \})|
}{
\sum_{m=1}^{M_0}h_{m0}(\{{\bf X}^{(z-1)}   \})
}\geq\delta$\\ 
\hspace{.6cm}{Calculate $\{{\bf X}^{(z)}\}$ by solving \eqref{(2-min-cor)}}\\
\hspace{.6cm}$z=z+1$\\
\hspace{0.2cm}\textbf{End (While)}\\
\hspace{0.2cm}{{\bf Return} $\{\mathbf{W}^{(\star)}\}$}\\
\hline 
\end{tabular}  
\end{table}
\singlespacing

\begin{figure}[t]
    \centering
    \begin{subfigure}[t]{0.24\textwidth}
        \centering
           \includegraphics[width=\textwidth]{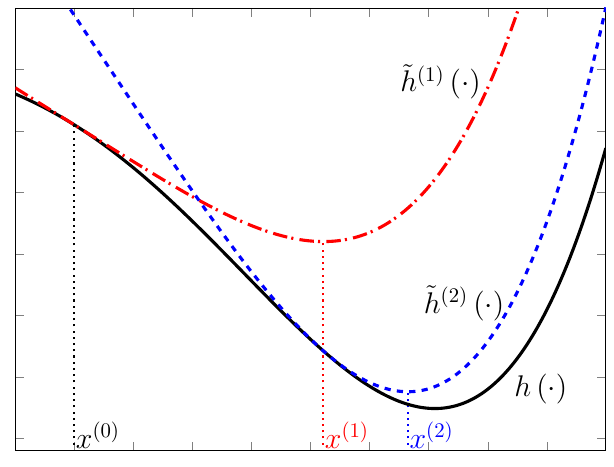}%{fig/mm-example.pdf}
        \caption{Minimization problem.}
    \end{subfigure}
\begin{subfigure}[t]{0.24\textwidth}
        \centering       \includegraphics[width=\textwidth]{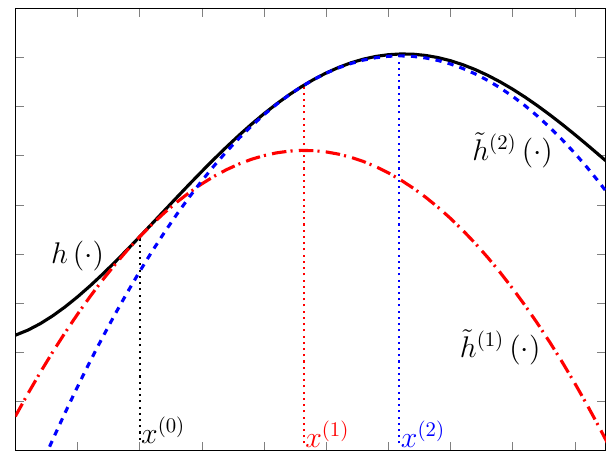}%{fig/mm-ex.pdf}
        \caption{Maximization problem.}
    \end{subfigure}%
    \caption{Examples of iterative MM-based algorithms.} 
    \label{Fig-mm} 
\end{figure}

\subsection{Maximization Problems}\label{sec+ii}
In this subsection, we aim for solving 
    \begin{align}\label{(1)}
        \max_{\{{\bf X}   \}\in \mathcal{X}  } & \sum_{m=1}^{M_0}h_{m0}(\{{\bf X}   \})
        &%\\
        \text{s.t.}&\sum_{m=1}^{M_i}h_{mi}(\{{\bf X}   \})\geq 0, \, \forall i,
    \end{align}
where $\{{\bf X}   \}$, and $h_{m,i}(\{{\bf X}   \})$ $\forall m,i,$ are defined as in \eqref{(1-min)}.  The OP in \eqref{(1)} is a non-convex FMP problem, not solvable by classical FMP solvers, such as Dinkelbach-based algorithms. In the following theorem, we reformulate \eqref{(1)} as a non-FMP problem to attain its stationary point.
\begin{theorem}\label{th1}
    A stationary point of \eqref{(1)} can be obtained by iteratively solving 
    \begin{subequations}\label{(2)}
    \begin{align}
        \max_{\{{\bf X}   \}\in \mathcal{X}, {\bf t}  } & \sum_{m=1}^{M_0}\tilde{h}_{m0}^{(z)}(\{{\bf X}\})
        \\
        \label{eq-3b}\text{\em s.t.}&\sum_{m=1}^{M_i}\tilde{h}_{mi}^{(z)}(\{{\bf X}\})\geq 0, 
        \\
\label{eq-3c}        & f_{mi}(\{{\bf X}\})-t_{mi}^2\geq 0, \, \forall m,i,
        \\
        \label{eq-3d}&
        t_{mi}\geq 0, \, \forall m,i,
    \end{align}
\end{subequations}
where $ {\bf t}=\{t_{mi} \}_{\forall mi}$ is a set of auxiliary optimization variables,  $z$ is the iteration index, and 
\begin{subequations}\label{(4)}
\begin{align}
\tilde{h}_{mi}^{(z)}(\{{\bf X}\})&=2a_{mi}^{(z)}t_{mi}- a_{mi}^{(z)^2}g_{mi}(\{{\bf X}\}),& \forall m,i,\\
    a_{mi}^{(z)}&=\frac{\sqrt{f_i(\{{\bf X}^{(z)}   \})}}{g_i(\{{\bf X}^{(z)}   \})}, & \forall m,i,
\end{align}
\end{subequations}
where $\{{\bf X}^{(z)}\}=\{{\bf X}_1^{(z)},{\bf X}_2^{(z)},\cdots,{\bf X}_N^{(z)}\}$.
\end{theorem}
\begin{proof}
    Please refer to Appendix \ref{app-2}.
\end{proof}
\subsubsection{When $f_{mi}$ is concave and $g_{mi}$ is convex $\forall m,i$}
In this case, \eqref{(2)} is convex, and the framework converges to an SP of \eqref{(1)} since the surrogate functions $\tilde{h}_{mi}^{(z)}(\{{\bf X}\})$ fulfill the three conditions in Appendix \ref{app-mm} for all $m,i$, as shown in the proof of Theorem \ref{th1}. Note that constraint \eqref{eq-3c} is convex in ${\bf t}$. Moreover, the OF of \eqref{(2)} and CFs \eqref{eq-3b} and \eqref{eq-3d} are linear in ${\bf t}$.

\subsubsection{General case}
In practical scenarios, the functions $f_{mi}(\{{\bf X}   \})$ and $g_{mi}(\{{\bf X}   \})$ are neither convex nor concave in $\{{\bf X}\}$. 
In this case, we can utilize the following lemma to compute an SP of \eqref{(1)}.

\begin{lemma}\label{lem-max-sur}
If $f_{mi}(\{{\bf X}   \})$ and $g_{mi}(\{{\bf X}   \})$ are non-convex and non-concave, a stationary point of \eqref{(1)} can be calculated by iteratively solving 
    \begin{subequations}\label{(2-sur)}
    \begin{align}
        \max_{\{{\bf X}   \}\in \mathcal{X}, {\bf t}  } & \sum_{m=1}^{M_0}\left(2a_{m0}^{(z)}t_{m0}- a_{m0}^{(z)^2}\tilde{g}_{m0}(\{{\bf X}\}) \right)
        \\
        \label{eq-3b-sur}\text{\em s.t.}&\sum_{m=1}^{M_i}\left(2a_{mi}^{(z)}t_{mi}- a_{mi}^{(z)^2}\tilde{g}_{mi}(\{{\bf X}\}) \right), \, \forall i,
        \\
\label{eq-3c-sur}        & \tilde{f}_{mi}(\{{\bf X}\})-t_{mi}^2\geq 0, \, \forall m,i,
        \\
        \label{eq-3d-sur}&
        t_{mi}\geq 0, \, \forall m,i,
    \end{align}
\end{subequations}
where $ {\bf t}$,  $z$, $\{{\bf X}^{(z)}\}$ and $a_{mi}^{(z)}$ are defined as in Theorem \ref{th1}.
Additionally, $\tilde{f}_{mi}(\{{\bf X}   \})$ is a concave \textit{lower bound} for $f_{mi}(\{{\bf X}   \})$, and $\tilde{g}_{mi}(\{{\bf X}   \})$ is a convex \textit{upper bound} for $g_{mi}(\{{\bf X}   \})$, satisfying the following conditions:
\begin{itemize}
    \item $\tilde{f}_{mi}(\{{\bf X}^{(z)}   \})=f_{mi}(\{{\bf X}^{(z)}   \})$ and  $\tilde{g}_{mi}(\{{\bf X}^{(z)}   \})=g_{mi}(\{{\bf X}^{(z)}   \})$, 
    \item $\tilde{f}_{mi}(\{{\bf X}   \})<f_{mi}(\{{\bf X}   \})$ and  $\tilde{g}_{mi}(\{{\bf X}   \})>g_{mi}(\{{\bf X}   \})$, 
    \item $\left. \frac{\partial \tilde{f}_{mi}(\{{\bf X}   \})}{\partial \{{\bf X}   \}}\right|_{\{{\bf X}   \}=\{{\bf X}^{(z)}\}}
    =\left. \frac{\partial {f}_{mi}(\{{\bf X}   \})}{\partial \{{\bf X}   \}}\right|_{\{{\bf X}   \}=\{{\bf X}^{(z)}\}}$ and  $\left. \frac{\partial \tilde{g}_{mi}(\{{\bf X}   \})}{\partial \{{\bf X}   \}}\right|_{\{{\bf X}   \}=\{{\bf X}^{(z)}\}}
    =\left. \frac{\partial {g}_{mi}(\{{\bf X}   \})}{\partial \{{\bf X}   \}}\right|_{\{{\bf X}   \}=\{{\bf X}^{(z)}\}}$, 
\end{itemize}
for all $m,i$ and $\{{\bf X}   \}\in \mathcal{X}$.
\end{lemma}
\begin{proof}
    Please refer to Appendix \ref{app-prof-lem2}.
\end{proof}
 {Similar to the framework in Section \ref{sec-ii-min}, we adopt an approach inspired by MM to solve \eqref{(1)}. In the general case, we first minorize $h_{mi}(\{{\bf X}   \})$ and then solve the corresponding surrogate OP in \eqref{(2-sur)} to compute an SP of \eqref{(1)}, as shown in Fig. \ref{Fig-mm}b.} Algorithm II summarizes our solution for \eqref{(1)} in the general case. Note that when $f_{mi}(\{{\bf X}   \})$ is a quadratic convex function, we can directly use the inequality in \eqref{eq-63}. In this case, there is no need to use the auxiliary variable $t_{mi}$. We provide a practical example of maximizing FFs with a quadratic numerator in Section \ref{sec-sinr-max}.
\doublespacing 
\begin{table}[htb]
\small
\begin{tabular}{l}
\hline 
 \textbf{Algorithm II}: Our framework to solve \eqref{(1)} in the general case.  \\
\hline
\hspace{0.2cm}{\textbf{Initialization}}\\
\hspace{0.2cm}Set $\delta$,  $z=1$,  $\{\mathbf{X}\}=\{\mathbf{X}^{(0)}\}$\\
\hline %\hline
\hspace{0.2cm}
\textbf{While} $
\frac{
\sum_{m=1}^{M_0}h_{m0}(\{{\bf X}^{(z)}   \})
-
\sum_{m=1}^{M_0}h_{m0}(\{{\bf X}^{(z-1)}   \})
}{
\sum_{m=1}^{M_0}h_{m0}(\{{\bf X}^{(z-1)}   \})
}\geq\delta$\\ 
\hspace{.6cm}{Calculate $\{{\bf X}^{(z)}\}$ by solving \eqref{(2-sur)}}\\
\hspace{.6cm}$z=z+1$\\
\hspace{0.2cm}\textbf{End (While)}\\
\hspace{0.2cm}{{\bf Return} $\{\mathbf{W}^{(\star)}\}$}\\
\hline 
\end{tabular} %\vspace{-.6cm}
\end{table}
\singlespacing

\subsubsection{Comparison with the Framework in \cite{shen2018fractional}}
The authors of \cite{shen2018fractional} proposed a quadratic transform to solve \eqref{(1)} iteratively. Upon using the framework in \cite{shen2018fractional} in each iteration,  \eqref{(1)} is transformed to 
\begin{subequations}\label{([10]-sur)}
    \begin{align}
        \max_{\{{\bf X}   \}\in \mathcal{X}, {\bf t}  } & \sum_{m=1}^{M_0}\left(2t_{m0}\sqrt{{f}_{m0}(\{{\bf X}\})}- t_{m0}^{2}{g}_{m0}(\{{\bf X}\}) \right)
        \\
        %\label{eq-3b-sur}
        \text{s.t.}&\sum_{m=1}^{M_i}\!\!\left(2t_{mi}\sqrt{{f}_{mi}(\{{\bf X}\})}- t_{mi}^{2}{g}_{mi}(\{{\bf X}\}) \right), \, \forall i,\!
        \\
        %\label{eq-3d-sur}
        &
        t_{mi}\geq 0, \, \forall m,i,
    \end{align}
\end{subequations}
which requires a joint optimization of $\{{\bf X}\}$ and ${\bf t}$. Unfortunately, \eqref{([10]-sur)} is non-convex even when ${f}_{mi}(\{{\bf X}\})$ and ${g}_{mi}(\{{\bf X}\})$ are linear in $\{{\bf X}\}$. To solve \eqref{([10]-sur)}, one should employ alternating optimization to separate the optimization of ${\bf X}$ and ${\bf t}$. That is, \eqref{([10]-sur)} is solved by optimizing ${\bf X}$, while ${\bf t}$ is kept fixed. Then ${\bf t}$ is updated by solving \eqref{([10]-sur)} when ${\bf X}$ is kept fixed. Hence, this approach requires solving a pair of OPs to update $\{{\bf X}\}$ in each iteration. By contrast, our framework requires solving only a single OP in each iteration. Additionally, in the general case where  ${f}_{mi}(\{{\bf X}\})$ is non-concave, it is challenging to obtain a suitable concave lower bound for $\sqrt{{f}_{mi}(\{{\bf X}\})}$. However, this issue is addressed in our framework, since Lemma \ref{lem-max-sur} requires only concave lower bounds for ${f}_{mi}(\{{\bf X}\})$.

Note that the framework in \cite{shen2018fractional} is proposed for solving only maximization OPs. However, we propose two frameworks for both maximization and minimization OPs. We are unaware of any other framework capable of solving the minimization FMP OPs that our framework in Section \ref{sec-ii-min} can solve. Additionally, we consider practical examples for MU-MIMO FBL systems, while \cite{shen2018fractional} focuses on different network scenarios when solving concrete maximization problems. 

\subsection{Discussions on the General Maximization and Minimization Problems}
In Section \ref{sec-ii-min} and Section \ref{sec+ii}, we consider $h_{mi}(\{{\bf X} \})\geq 0$; however, the frameworks proposed in this paper can also be applied to OPs associated with negative functions. A practical example of such FMP OPs is constituted by multi-objective OPs, targeting to maximize a term while minimizing another one, yielding a format of $h_{01}(\{{\bf X}   \})-h_{00}(\{{\bf X}   \})$ for the OF of the OP. It can be the case, for instance, when we aim for maximizing the EE, while simultaneously minimizing latency. When maximizing an FF utility function in the format of $h_{01}(\{{\bf X}   \})-h_{00}(\{{\bf X}   \})$, we can apply the approach in Section \ref{sec+ii} to find a lower bound for  $h_{01}(\{{\bf X}   \})$ and the approach in Section \ref{sec-ii-min} to calculate an upper bound for $h_{00}(\{{\bf X}   \})$. Due to space restriction, we skip the details, as it is a straightforward extension of the solutions provided in Section \ref{sec-ii-min} and Section \ref{sec+ii}. 

Our framework can also solve FMP OPs with the format of 
\begin{align}\label{eq-new-1}
        \min_{\{{\bf X}   \}\in \mathcal{X}  } & \sum_{m=1}^{M_0}h_{m0}(\{{\bf X}   \})
        &%\\
        \text{s.t.}&\sum_{m=1}^{M_i}h_{mi}(\{{\bf X}   \})\geq 0, \, \forall i.
\end{align}
To solve \eqref{eq-new-1}, we can leverage the approach in Section \ref{sec-ii-min} to deal with the OF, while using the approach in Section \ref{sec+ii} to handle the CFs.
Alternatively, the  framework can also solve FMP OPs having the following format
\begin{align}\label{eq-new-2}
        \max_{\{{\bf X}   \}\in \mathcal{X}  } & \sum_{m=1}^{M_0}h_{m0}(\{{\bf X}   \})
        &%\\
        \text{s.t.}&\sum_{m=1}^{M_i}h_{mi}(\{{\bf X}   \})\leq \eta_i, \, \forall i.
\end{align}
To this end, we can employ the approach in Section \ref{sec+ii} to find a surrogate function for the OF, while utilizing the approach in Section \ref{sec-ii-min} to derive a surrogate function for the CFs.

In the aforementioned FMP OPs, we assume that the OF and CFs are linear in $h_{mi}(\{{\bf X}\})$ $\forall m,i$. However, our optimization framework can also solve OPs in which the OF and CFs are non-linear in $h_{mi}(\{{\bf X}\})$, i.e., in a format of $\phi_{mi}(h_{mi}(\{{\bf X}\}))$, where $\phi_{mi}(\cdot)$ can be even non-convex or non-concave. We will discuss this issue in Section \ref{sec=iii} and Section \ref{sec=iii-max}, providing a pair of practical examples. Additionally, we note that our framework can also solve products of FFs as
\begin{equation}
   \prod_{i} \frac{f_{mi}(\{{\bf X}   \})}{g_{mi}(\{{\bf X}   \})}, \,\, \forall m.
\end{equation}
This type of functions include, for instance, geometric mean optimization as we will discuss in the next sections. Finally, we note that a key step in our optimization framework is to find adequate surrogate functions for $f_{mi}(\{{\bf X}   \})$ and/or $g_{mi}(\{{\bf X}   \})$, which can be challenging in many practical scenarios. In Appendix \ref{app-1}, we provide inequalities to expedite calculating efficient surrogate functions, especially for  MIMO systems.

\section{Practical Applications of FMP Minimization Problems}\label{sec=iii}
The unified optimization framework of Section \ref{secii-new} can solve a large family of FMP problems. In this section, we provide various examples of such FMP minimization problems, focusing on MU-MIMO systems relying on FBL coding. Particularly,  we minimize the sum delay, geometric mean of delays, channel dispersion, and MSE. To provide a detailed solution, we consider a single-cell MIMO BC in this section. 
Section \ref{sec=iv} describes the applications of our framework in other MU-MIMO systems. 
We provide a few practical examples of the FMP minimization problems, solvable by our optimization framework in Table \ref{tab:minOPs}.
\doublespacing
\begin{table}[t]
    \centering    \footnotesize
    \caption{Practical applications of the objective functions solvable by our framework in Section \ref{sec-ii-min}.}
    \begin{tabular}{|l||l|}
    \hline
     Sum delay& $ \sum_{k}\frac{L_{k}}{r_{k} }$
     \\     
     \hline
     Geometric mean of delay&  $\prod_{k}\left(\frac{L_{k}}{r_{k} }\right)^{1/K} $
     \\     
     \hline
     Maximum MSE&$\max_k\left\{\text{Tr}\left({\bf I}-{\bf D}_k^{-1}{\bf S}_k \right)\right\}$
     \\     
     \hline
     Sum MSE&$\sum_k\text{Tr}\left({\bf I}-{\bf D}_k^{-1}{\bf S}_k \right)$
     \\     
     \hline
    \end{tabular} 
    \label{tab:minOPs}
\end{table}
\normalsize
\singlespacing

\subsection{MIMO Broadcast Channels}
\subsubsection{Network scenario} In the practical examples of this section, we study a MIMO BC, consisting of a single BS having $N_{BS}$ transmit antennas (TAs), serving $K$ users each equipped with $N_u$ receive antennas (RAs). Moreover, we treat interference as noise (TIN) and employ FBL coding having a codeword length $n$ at the BS to make the rates more general than the classical Shannon rates. In the following, we briefly state the achievable rate and EE of users, and refer the reader to, e.g., \cite{soleymani2024optimization} for more in-depth discussions.

\subsubsection{Rate and EE expressions} Upon using the normal approximation (NA), the achievable rate of user $k$ is given by \cite[Lemma 1]{soleymani2024optimization}
\begin{equation}\label{eq-rate}
r_{k}=
\underbrace{\log \left|{\bf I} +{\bf D}^{-1}_{k}{\bf S}_{k} \right|}_{r_{S,k}=\text{\footnotesize Shannon Rate}}
-Q^{-1}(\epsilon)\sqrt{\frac{v_{k}}{n}},
\end{equation}
where  {$Q^{-1}$ is the inverse Q-function for Gaussian signals}, $v_{k}$ is the channel dispersion of user $k$, ${\bf S}_{k}={\bf H}_{k}{\bf W}_{k}({\bf H}_{k}{\bf W}_{k})^H$ is the covariance matrix of the desired signal at user ${k}$, and
 ${\bf D}_{k}$ 
 is the covariance  matrix of the interfering signals plus noise, given by
\begin{equation}
{\bf D}_{k}={\bf C}_{k}
+
\sum_{j=1,j\neq k}^K{\bf H}_{k}{\bf W}_{j}({\bf H}_{k}{\bf W}_{j})^H,
\end{equation}
where ${\bf C}_{k} $ is the covariance of the additive noise at user $k$, ${\bf H}_{k}$ is the channel between the BS and  user ${k}$, while ${\bf W}_{j}$ is the beamforming matrix, corresponding to the signal containing the data of user ${k}$. An achievable channel dispersion for MU-MIMO systems with Gaussian signals is \cite{scarlett2016dispersion, soleymani2024optimization}
\begin{equation}\label{eq-ch-dis}
    v_{k}=2\text{Tr}({\bf S}_k({\bf D}_k+{\bf S}_k)^{-1}),
\end{equation}
which is an FF of $\{{\bf W}\}=\{ {\bf W}_j\}_{\forall j}$.

The EE of user $k$ is \cite{buzzi2016survey}
\begin{equation}\label{eq-ee}
    e_k=\frac{r_k}{P_s+\eta \text{Tr}({\bf W}_k{\bf W}_k^H) },
\end{equation}
where $P_s$ is the static power required for transmitting data to user $k$, and $\eta$ is the power efficiency of the power amplifier at the BS. Furthermore, the GEE of the network is \cite{buzzi2016survey}
\begin{equation}\label{eq-gee}
    g=\frac{\sum_k r_k}{\sum_k \left[P_s+\eta \text{Tr}({\bf W}_k{\bf W}_k^H)\right] }.
\end{equation}
Finally, the transmission delay of user $k$ is given by \cite{bai2020latency, li2023min}
\begin{equation}\label{eq-delay-18}
    d_k=\frac{L_k}{r_k},
\end{equation}
where $L_k$ is the packet length of user $k$.

The rate function in \eqref{eq-rate} is non-concave and non-convex. A concave lower-bound for the rates was obtained in \cite[Lemma 5]{soleymani2024optimization}. Here, we restate the bound for the ease of the reader. %\vspace{-.1cm}
\begin{lemma}[\!\!\cite{soleymani2024optimization}]\label{lem-1}
A concave lower bound for $r_{k}$ is %given by
\begin{multline}%{equation}
\label{eq24}
r_{k}\geq \tilde{r}_{k}= a_{k}
+2\sum_{j}\mathfrak{R}\left\{\text{{\em Tr}}\left(
{\bf A}_{kj}\mathbf{W}_{j}^H
\bar{\mathbf{H}}_{k}^H\right)\right\}
\\
-
\text{{\em Tr}}\left(
{\bf B}_{k}(\mathbf{H}_{k}\mathbf{W}_{k}\mathbf{W}_{k}^H\mathbf{H}_{k}^H+\mathbf{D}_{k})
\right),
\end{multline}%{equation}
where the coefficients $a_{k}$, ${\bf A}_{kj}$, and ${\bf B}_{k}$ are given by \cite[Lemma 5]{soleymani2024optimization}.
\end{lemma} %\vspace{-.5cm}

\begin{table}[t]
\footnotesize
%\scriptsize
    \centering
    \caption{List of most frequently used notations in Section \ref{sec=iii}.}
    \begin{tabular}{|l|l|}
    \hline
        ${\bf H}_{k}$ & Channel between the BS and user $k$ \\
        \hline
         ${\bf W}_{k}$ & Beamforming matrix for data transmission to user ${k}$
         \\
         \hline
         $r_{k}$ & Rate of user ${k}$, given by \eqref{eq-rate}
         \\
         \hline
         $r_{S,k}$ & Shannon rate of user ${k}$, given by \eqref{eq-rate}
         \\
         \hline
         $e_{k}$ & EE of user ${k}$, given by \eqref{eq-ee}
         \\
         \hline
         $g$& Global EE of the network, given by \eqref{eq-gee}
         \\
         \hline
         $\tilde{r}_{k}$ & Concave lower bound for $r_{k}$, given by Lemma \ref{lem-1}
         \\
         \hline
         $L_k$ & Packet length for user $k$
         \\
         \hline
         $p_k$& Power consumption for data TX to user $k$, given by \eqref{eq-pow}
         \\
         \hline
         $v_k$& Achievable channel dispersion for user $k$, given by \eqref{eq-ch-dis}
         \\
         \hline
         $\zeta$&Tradeoff between GEE and sum rate, given by \eqref{eq-11}
         \\
         \hline
         $\zeta_k$&Tradeoff between EE and rate of user $k$, given by \eqref{eq-see}
         \\
         \hline
          $\xi_k$&MSE of user $k$, given by \eqref{eq-mse}
         \\
         \hline
    \end{tabular} %\vspace{-.5cm}
    \label{tab:1}
\end{table}
\normalsize

%\vspace{-.1cm}
\subsection{Sum Delay Minimization}\label{sec-sum-delay}
The minimization of the sum transmission delay can be formulated as 
\begin{subequations}\label{eq-sdm}
    \begin{align}
        \min_{\{{\bf W}\} }\, & \sum_k\frac{L_k}{r_k}
        \\
        \text{s.t.}\,\,& \sum_k\text{Tr}({\bf W}_k{\bf W}^H_k)\leq P,
        \label{eq-13-b}
    \end{align}
\end{subequations}
where $P$ is the power budget of the BS. Upon leveraging Lemma \ref{cor-min} and Lemma \ref{lem-1}, we can attain an SP of  \eqref{eq-sdm} by iteratively solving
%\begin{subequations}
\begin{align}\label{eq-sdm-sur}
        \min_{\{{\bf W}\},{\bf t} }\, & \sum_k\frac{L_k}{t_k}
        &
        \text{s.t.}\,\,& \eqref{eq-13-b},\,\, \tilde{r}_k\geq t_k>0,\,\, \forall k,  
        %\\        &&&\eqref{eq-13-b}.
    \end{align}
%\end{subequations}
where ${\bf t}=\{t_1,t_2,\cdots,t_K\}$ is the set of auxiliary variables. Note that in \eqref{eq-sdm}, $L_k$ is assumed to be constant. However, our framework can also solve OPs in which $L_k$ is an optimization variable. It happens when a joint design of the physical (PHY) layer and higher layers, e.g., the multiple access (MAC) or network, is studied. Moreover, in applications to mobile-edge computing, the number of bits offloaded to a computing center via a communication link is an optimization variable \cite{bai2020latency, li2023min}. In these applications, the resultant latency-minimization problem falls into FMP, solvable by our algorithm in Section \ref{sec-ii-min}. To avoid unnecessary complications of the practical examples, we leave it for future research.

 %\vspace{-.5cm}
\subsection{Geometric Mean of Delays}\label{sec-geo-de}
The geometric mean of the transmission delays is
\begin{equation}
    \bar{d}=\prod_{k} d_k^{\frac{1}{K}}=\prod_{k}\left(\frac{L_k}{r_k}\right)^{\frac{1}{K}},
\end{equation}
and minimizing $\bar{d}$ can be written as
%\begin{subequations}
    \begin{align}\label{eq-gmd}
        \min_{\{{\bf W}\} }\, & \prod_{k}\left(\frac{L_k}{r_k}\right)^{\frac{1}{K}}
        &
        \text{s.t.}\,\,& \eqref{eq-13-b}.
    \end{align}
%\end{subequations}
 {The OF of \eqref{eq-gmd} is a multiplication of FFs, making \eqref{eq-gmd} more complex than \eqref{eq-sdm}. To solve \eqref{eq-gmd}, we first employ the upper bound in Lemma \ref{lem-gm}, transferring \eqref{eq-gmd} to the format of \eqref{(1-min)}. Note that the upper bound in Lemma \ref{lem-gm} fulfills the three conditions of MM algorithms mentioned in Lemma \ref{lem-mm}, hence ensuring  convergence to an SP of \eqref{eq-gmd}.}
Upon employing the inequality in Lemma \ref{lem-gm}, we have
%\begin{subequations}
    \begin{align}\label{eq-gmd2}
        \min_{\{{\bf W}\} }\, & \sum_{k}\alpha_k \frac{L_k}{r_k}
        &
        \text{s.t.}\,\,& \eqref{eq-13-b},
    \end{align}
%\end{subequations}
where $\alpha_k$ is a constant coefficient, given by
\begin{equation}
    \alpha_k\!=\!\frac{1}{K}\!\left(\!\frac{L_i}{r_k(\{{\bf W}^{(z-1)} \}) }\!\right)^{\frac{1-K}{K}}\prod_{i\neq k}\!\left(\!\frac{L_i}{r_i(\{{\bf W}^{(z-1)} \})}\right)^{\frac{1}{K}}\!.\!
\end{equation}
Now we can leverage Lemma \ref{cor-min} and Lemma \ref{lem-1} to attain an SP of \eqref{eq-gmd} by iteratively solving 
%\begin{subequations}
\begin{align}\label{eq-gmd3}
    \min_{\{{\bf W}\},{\bf t} }\, & \sum_k\frac{\alpha_k L_k}{t_k}
        &
        \text{s.t.}\,\,& \eqref{eq-13-b},\,\, \tilde{r}_k\geq t_k>0,\,\, \forall k.
\end{align}
Note that the geometric mean can be considered as a fairness metric in multi-user systems \cite{yu2021maximizing}.

%\vspace{-.5cm}
\subsection{Minimizing Maximum or Sum Mean Square Error} 
Upon using an MMSE estimator, the MSE of user $k$ is \cite{shi2011iterative}
\begin{equation}\label{eq-mse}
   \xi_k =\text{Tr}\left({\bf I}-{\bf D}_k^{-1}{\bf S}_k \right),
\end{equation}
which has a fractional structure in $\{{\bf W}\}$. We consider a pair of minimization problems, having an MSE-related metric as the OF: i) minimizing the sum MSE, ii) minimizing the maximum (min-max) MSE. To solve these OPs, we employ the inequality in Lemma \ref{lem-3} to calculate a convex upper bound for $\xi_k$ in the following lemma.
\begin{lemma}\label{mse-up}
For all feasible $\{{\bf W}\}$ and $\{{\bf W}^{(z)} \}$, the following inequality holds:
\begin{multline}%{equation}  
\label{eq-mse-up}
\xi_k =\text{\em Tr}\left({\bf I}-{\bf D}_k^{-1}{\bf S}_k \right)\leq \tilde{\xi}_k
=
\text{\em Tr}\left({\bf I}+ \tilde{\bf D}^{-1}_k\tilde{\bf S}_k\tilde{\bf D}^{-1}_k{\bf D}_k\right)
\\-
2\mathfrak{R}\left\{\text{\em Tr}\left(\tilde{\bf D}^{-1}_k{\bf H}_k{\bf W}^{(z)}_k{\bf W}^H_k{\bf H}^H_k\right)\right\}
,
\end{multline}%{equation}
where $\tilde{\bf D}_k={\bf D}_k(\{{\bf W}^{(z)}\})$ and $\tilde{\bf S}_k={\bf S}_k(\{{\bf W}^{(z)}\})$. Note that $\tilde{\xi}_k$ is quadratic in $\{{\bf W}\}$.
\end{lemma}
\subsubsection{Minimizing sum MSE}
The sum MSE minimization can be formulated as 
\begin{align}\label{eq-sum-mse}
\min_{\{{\bf W}\}}  &\sum_k   \xi_k& \text{s.t}\,\,\,&\eqref{eq-13-b}.
\end{align}
Upon using Lemma \ref{mse-up}, we can attain an SP of \eqref{eq-sum-mse} by iteratively solving
\begin{align}\label{eq-sum-mse-sur}
\min_{\{{\bf W}\}}  &\sum_k   \tilde{\xi}_k
& \text{s.t}\,\,\,&\eqref{eq-13-b}.
\end{align}
\subsubsection{Max-min MSE}
We can formulate the max-min MSE as
\begin{align}\label{eq-max-min-mse}
\min_{\{{\bf W}\}}  &\max_k\{\xi_k\}& \text{s.t}\,\,&\eqref{eq-13-b}.
\end{align}
Leveraging Lemma \ref{mse-up}, we can compute an SP of \eqref{eq-max-min-mse} by iteratively solving 
\begin{align}\label{eq-max-min-mse-sur}
\min_{\{{\bf W}\}}  &\max_k\{\tilde{\xi}_k\}
& \text{s.t}\,\,&\eqref{eq-13-b}.
\end{align}

 %\vspace{-.5cm}
\subsection{Optimizing Channel Dispersion}\label{sec-ch-dis}
The channel dispersion is also an FF of $\{{\bf W}\}$, expressed as
\begin{equation}%\label{eq-ch-dis}
    v_{k}=2\text{Tr}({\bf S}_k({\bf D}_k+{\bf S}_k)^{-1}).
\end{equation}
All OPs associated with FBL coding require an optimization over the channel dispersion. In this subsection, we provide a convex upper bound for $v_k$, which can be used for minimizing $v_k$. Note that 
 $r_k$ is decreasing in $v_k$, which motivates minimizing $v_k$. 
To find the bound, we rewrite $v_k$ as follows
\begin{equation}%\label{eq-ch-dis}
    v_{k}=2\text{Tr}({\bf I} -{\bf D}_k({\bf D}_k+{\bf S}_k)^{-1}).
\end{equation}
Upon employing Lemma \ref{lem-3}, we have
\begin{multline}%{equation}  
\label{eq-mse-up}
v_k =2\text{Tr}({\bf I} -{\bf D}_k({\bf D}_k+{\bf S}_k)^{-1})
\leq \tilde{v}_k
\\
=
\text{Tr}\left({\bf I}+ (\tilde{\bf D}_k+\tilde{\bf S}_k)^{-1}\tilde{\bf D}_k(\tilde{\bf D}_k+\tilde{\bf S}_k)^{-1}({\bf D}_k+{\bf S}_k)\right)
\\-
2\sum_{j\neq k}\mathfrak{R}\left\{\text{Tr}\left((\tilde{\bf D}_j+\tilde{\bf S}_j)^{-1}{\bf H}_j{\bf W}^{(z)}_j{\bf W}^H_j{\bf H}^H_j\right)\right\}
,
\end{multline}%{equation}
where $\tilde{\bf D}_k$ and $\tilde{\bf S}_k$ are defined as in Lemma \ref{mse-up} for all $k$.
The upper bound $\tilde{v}_k$ is quadratic and convex in $\{{\bf W} \}$.

 %\vspace{-.3cm}
\section{Practical Applications of FMP Maximization Problems}\label{sec=iii-max}
In this section, we provide four examples of practical FMP maximization problems with applications to MU-MIMO systems using FBL coding. Specifically, we maximize the spectral-energy efficiency (SEE) tradeoff metrics, weighted sum EE (WSEE), and the geometric mean of EE (GMEE), as mentioned in Table \ref{tab:maxOPs}. Similar to Section \ref{sec=iii}, we consider a single-cell MIMO BC in this section.  
\doublespacing
\begin{table}[t]
    \centering    \footnotesize
    %\scriptsize
    \caption{Practical applications of the objective functions solvable by our framework in Section \ref{sec+ii}.}
    \begin{tabular}{|l||l|}
    \hline
     GEE \& sum rate tradeoff    & $\alpha \sum_{k}r_{k}+ \frac{(1-\alpha)\sum_{k}r_{k}}{\sum_{k}\left[P_s+\eta \text{Tr}({\bf W}_{k}{\bf W}_{k}^H )\right] }$  
     \\ 
     \hline
     Max-min EE \& SE tradeoff    & $\min_{k} \left\{ \alpha_{k} r_{k}+ \frac{(1-\alpha_{k})r_{k}}{P_s+\eta \text{Tr}({\bf W}_{k}{\bf W}_{k}^H ) }\right\}$  
     \\
     \hline
       Weighted sum EE  &  $ \sum_{k}\frac{\alpha_{k}r_{k}}{P_s+\eta \text{Tr}({\bf W}_{k}{\bf W}_{k}^H ) }$
       \\
       \hline
       Geometric mean of EE& $ \prod_{k} \left(\frac{r_{k}}{P_s+\eta \text{Tr}({\bf W}_{k}{\bf W}_{k}^H ) }\right)^{1/k}$
       \\     
     \hline
    \end{tabular} %\vspace{-.5cm}
    \label{tab:maxOPs}
\end{table}
\normalsize
\singlespacing

 %\vspace{-.6cm}
\subsection{Spectral-energy-efficiency Tradeoff}\label{sec-iv-a}
We consider a pair of performance metrics to evaluate the tradeoff between the SE and EE in this subsection.

\subsubsection{Sum Rate and GEE Tradeoff}
A metric for studying the SEE tradeoff is given by \cite{zappone2023rate, aydin2017energy, you2020energy, you2020spectral}
\begin{equation}\label{eq-11}
   \zeta=\alpha\sum_k r_k+ \frac{(1-\alpha)\sum_k r_k}{\sum_k[P_s+\eta \text{Tr}({\bf W}_k{\bf W}_k^H)] } ,
\end{equation}
where $0\leq\alpha\leq 1$ and $(1-\alpha)$ are the weights corresponding to the SE and EE, respectively.
We can formulate the maximization of $\zeta$ as follows
%\begin{subequations}
    \begin{align}\label{opt-13-see}
        \max_{\{{\bf W}\} }\,\, & \alpha\sum_k r_k
        +
        \frac{(1-\alpha)\sum_k r_k}{\sum_k[P_s+\eta \text{Tr}({\bf W}_k{\bf W}_k^H)] } 
        &
        \text{s.t.}\,\,& 
        \eqref{eq-13-b}.
    \end{align}
%\end{subequations}
Leveraging Lemma \ref{lem-max-sur} and Lemma \ref{lem-1} yields the convex OP
\begin{subequations}
    \begin{align}
        \max_{\{{\bf W}\},t }\, & 
        \alpha\sum_k \tilde{r}_k 
        \!+ 
        (1\!-\!
        \alpha)\!
        \left(
        \!\!2\beta^{(z)}t- \beta^{(z)^2}
        {\sum_kp_k(\{{\bf W}\}) } 
        \!\!\right)
        \\
        \text{s.t.}\,\,& 
        \eqref{eq-13-b},\,\,\sum_k\tilde{r}_k -t^2\geq 0,
    \end{align}
\end{subequations}
where $\beta^{(z)}=\frac{\sqrt{\sum_k r_k({\bf W}^{(z)})}}{\sum_kp_k({\bf W}^{(z)})}$ and
\begin{equation}\label{eq-pow}
    p_k({\bf W})=P_s+\eta \text{Tr}({\bf W}_k{\bf W}_k^H).
\end{equation}

\subsubsection{Max-min Rate and EE Tradeoff}
Another SEE tradeoff metric is based on the individual rate and EE of each user given by  
\begin{equation}\label{eq-see}
   \zeta_k=\alpha_k r_k+ \frac{(1-\alpha_k) r_k}{P_s+\eta \text{Tr}({\bf W}_k{\bf W}_k^H) },
\end{equation}
where $0\leq \alpha_k\leq 1$ and $1-\alpha_k$ are coefficients reflecting the weights of the SE and EE for user $k$, respectively. 
Maximizing the minimum of $\zeta_k$ for all $k$ can be expressed as
%\begin{subequations}
    \begin{align}
        \max_{\{{\bf W}\} }\, & \min_k\left\{
        \alpha_k r_k
        +\frac{(1-\alpha_k) r_k}{P_s+\eta \text{Tr}({\bf W}_k{\bf W}_k^H) }
        \right\}
        &
        \text{s.t.}\,\,& 
        \eqref{eq-13-b}.
    \end{align}
%\end{subequations}
Upon utilizing Lemma \ref{lem-max-sur} and Lemma \ref{lem-1}, we obtain the convex OP 
\begin{subequations}\label{(18)}
    \begin{align}
        \max_{\{{\bf W}\},{\bf t} }\, & \min_k
        \left\{
        \alpha_k \tilde{r}_k
        +(1-\alpha_k)
        \left(2\beta_k^{(z)} t_k-\beta_k^{(z)^2}p_k \right)
        \right\}
        \\
        \text{s.t.}\,\,& \
       \eqref{eq-13-b},\,\,
        \tilde{r}_k-t_k^2\geq 0,\,\,\, \forall k,
    \end{align}
\end{subequations}
where $\beta_k^{(z)}=\frac{\sqrt{r_k(\{{\bf W}^{(z)}\})}}{p_k(\{{\bf W}^{(z)}\})}$, $p_k$ is given by \eqref{eq-pow}, and ${\bf t}=\{t_1, t_2,\cdots,t_K\} $ is the set of auxiliary optimization variables. 

\subsection{Weighted Sum EE}\label{sec-sum-ee}
The weighted sum EE can be formulated as
    \begin{align}\label{eq-sum-ee}
        \max_{\{{\bf W}\} } & \sum_k\!
        \frac{\alpha_k r_k}{P_s+\eta \text{Tr}({\bf W}_k{\bf W}_k^H) }
        &
        \text{s.t.}\,& \eqref{eq-13-b},r_k\geq r_k^{th}, \, \forall k,\!\!
    \end{align}
where $\alpha_k$s are the corresponding weights, and $r_k^{th}$ is the minimum operational rate for user $k$ to ensure a specific quality of service (QoS). 
To find an SP of \eqref{eq-sum-ee}, we utilize Lemma \ref{lem-max-sur} and Lemma \ref{lem-1}, yielding the convex OP
\begin{subequations}\label{eq-45-sumEE}
    \begin{align}
        \max_{\{{\bf W}\},{\bf t} }\, & \sum_k\!\alpha_k\left(
        2\beta_k^{(z)} t_k-\beta_k^{(z)^2}{(P_s+\eta \text{Tr}({\bf W}_k{\bf W}_k^H)) }
        \right)\!\!
        \\
        \text{s.t.}\,\,& \eqref{eq-13-b},\,\,
        \alpha_k \tilde{r}_k-t_k^2\geq 0,\,\,\tilde{r}_k\geq r_k^{th}, \,\,\, \forall k,
    \end{align}
\end{subequations}
where $\beta_k^{(z)}=\frac{\sqrt{r_k(\{{\bf W}^{(z)}\})}}{p_k(\{{\bf W}^{(z)}\})}$, and ${\bf t} $ is defined as in \eqref{(18)}.

 %\vspace{-.4cm}
\subsection{Geometric Mean of EE Functions}
The geometric mean (GM) of EE is defined as
\begin{equation}
     \bar{e}=\prod_{k} \left(\frac{r_{k}}{P_s+\eta \text{Tr}({\bf W}_{k}{\bf W}_{k}^H ) }\right)^{1/k}.
\end{equation}
Hence, the maximization of the GM of the EE can be written as 
\begin{subequations}\label{eq-gm-ee}
    \begin{align}
        \max_{\{{\bf W}\} }\, & \prod_{k} \left(\frac{r_{k}}{P_s+\eta \text{Tr}({\bf W}_{k}{\bf W}_{k}^H ) }\right)^{1/k}
        \\
        \text{s.t.}\,\,& \eqref{eq-13-b},r_k\geq r_k^{th}, \,\,\, \forall k.
    \end{align}
\end{subequations}
Maximizing $x^{1/K} $ is equivalent to maximizing $x^2$ for $x\geq 0$. Hence, the solution of \eqref{eq-gm-ee} is equivalent to the solution of
\begin{subequations}\label{eq-gm-ee2}
    \begin{align}
        \max_{\{{\bf W}\} }\, & \prod_{k} \left(\frac{r_{k}}{P_s+\eta \text{Tr}({\bf W}_{k}{\bf W}_{k}^H ) }\right)^{2}
        \\
        \text{s.t.}\,\,& \eqref{eq-13-b},r_k\geq r_k^{th}, \,\,\, \forall k.
    \end{align}
\end{subequations}
Upon utilizing Lemma \ref{lem-max-sur} and the lower bounds in Lemma \ref{lem-1} and Lemma \ref{lem-gm2}, we can obtain 
\begin{subequations}\label{eq-gm-ee3}
    \begin{align}
         \max_{\{{\bf W}\},{\bf t} }\, & \sum_k\alpha_k\left(
        2\beta_k^{(z)} t_k-\beta_k^{(z)^2}{(P_s+\eta \text{Tr}({\bf W}_k{\bf W}_k^H)) }
        \right)\!\!
        \\
        \text{s.t.}\,\,& \eqref{eq-13-b},\,\,
         \tilde{r}_k-t_k^2\geq 0,\,\,\tilde{r}_k\geq r_k^{th}, \,\,\, \forall k,
    \end{align}
\end{subequations}
where $\beta_k^{(z)}=\frac{\sqrt{r_k(\{{\bf W}^{(z)}\})}}{p_k(\{{\bf W}^{(z)}\})}$, and $\alpha_k$ is a coefficient, given by 
\begin{equation}
    \alpha_k=2{e}_k\left(\{{\bf W}^{(z)}\}\right)\prod_{i\neq k}{e}_i^{2}\left(\{{\bf W}^{(z)}\}\right),\,\,\forall k.
\end{equation}

 %\vspace{-.7cm}
\subsection{Weighted-sum SINR Maximization}\label{sec-sinr-max}
The focus of this paper is on MU-MIMO systems supporting multiple-stream data transmission per users. However, for the sake of completeness, we also solve the maximization of weighted-sum SINR in MIMO systems, when a single-stream data transmission is utilized. In this case, we have to optimize over the beamforming vectors $\{{\bf w}\}=\{{\bf w}_1,{\bf w}_2,\cdots, {\bf w}_K\} $,  instead of the matrices $\{{\bf W} \}$, where ${\bf w}_k$ is the beamforming vector, corresponding to user $k$. Thus, the SINR of user $k$ is
\begin{equation}
    \gamma_k=\frac{{\bf w}_k^H{\bf H}_k^H{\bf H}_k{\bf w}_k}{\sigma^2+\sum_{j\neq k}{\bf w}_j^H{\bf H}_k^H{\bf H}_k{\bf w}_j},
\end{equation}
where  $\sigma^2$ is the noise variance at the receiver of user $k$. 

The weighted-sum SINR maximization can be formulated as
\begin{align}\label{eq-sinr-max}
    \max_{\{{\bf w}\} }&\sum_k \alpha_k \gamma_k&\text{s.t.}&\sum_k{\bf w}_k{\bf w}_k^H\leq P,
\end{align}
where $\alpha_k$ is the corresponding SINR weight for user $k$. To obtain an SP of \eqref{eq-sinr-max}, we employ the inequality in \eqref{eq-63} to calculate a surrogate function for $\gamma_k$ as
\begin{multline}\label{eq-45-sinr}
    \gamma_k\geq \tilde{\gamma}_k=\frac{2\mathfrak{R}\left\{({\bf H}_k{\bf w}_k^{(z)})^H{\bf H}_k{\bf w}_k\right\}}{\sigma^2+\sum_{j\neq k}|{\bf h}_k{\bf w}_j^{(z)}|^2}
    \\-\frac{|{\bf h}_k{\bf w}_k^{(z)}|^2}{\left(\sigma^2+\sum_{j\neq k}|{\bf h}_k{\bf w}_j^{(z)}|^2\right)^2}\left(\sigma^2+\sum_{j\neq k}|{\bf h}_k{\bf w}_j|^2\right),
\end{multline}
which is quadratic and concave. Then, we attain an SP of \eqref{eq-sinr-max} by iteratively solving
\begin{align}\label{eq-sinr-max2}
    \max_{\{{\bf w}\} }&\sum_k \alpha_k \tilde{\gamma}_k&\text{s.t.}&\sum_k{\bf w}_k{\bf w}_k^H\leq P.
\end{align}
Note that the bound in \eqref{eq-45-sinr} can be employed to obtain an SP of other OPs with SINRs such as the max-min SINR of users. 

%\vspace{-.5cm}
\section{Extension to RIS-aided Systems }\label{sec=iv}
 In RIS-aided systems, we can optimize the system performance not only through the beamforming matrices $\{{\bf W} \}$ at the BS, but also through the channels by modifying the RIS coefficients. In this section, we extend our solutions in Section \ref{sec=iii} and Section \ref{sec=iii-max} to MU-MIMO RIS-aided BCs relying on FBL coding. To this end, we employ an alternating optimization (AO) approach, which is a typical technique of developing resource allocation schemes for multi-user RIS-assisted systems. More particularly, we separate the optimization of $\{{\bf W} \}$ and RIS coefficients ${\bf\Theta}$ at each iteration. Indeed, we first optimize over $\{{\bf W}\}$, when ${\bf\Theta}$ is kept fixed at ${\bf\Theta}^{(z-1)}$. Then we alternate and optimize over ${\bf\Theta}$, when $\{{\bf W} \}$ is kept fixed at $\{{\bf W}^{(z-1)} \}$. Optimization of $\{{\bf W} \}$ is similar to the solutions detailed in Section \ref{sec=iii} and Section \ref{sec=iii-max}. Thus, we only consider the optimization of ${\bf\Theta}$ in this section. To this end, a concave lower bound of the rates in ${\bf\Theta}$ is needed similar to when optimizing over $\{{\bf W} \}$. A concave lower bound for $\hat{r}_k$ is provided in \cite[Corollary 1]{soleymani2024optimization}. To avoid redundancy, we do not restate them here.

We utilize the channel model in \cite{pan2020multicell} for a MIMO RIS-aided link as
%\begin{equation}
    ${\bf H}_k= {\bf G}_k{\bf\Theta} {\bf G}+\tilde{{\bf G}}_k$,
%\end{equation}
where $\tilde{{\bf G}}_k$ is the channel between the BS and user $k$, ${\bf G}$ is the channel between the RIS and the BS, ${\bf G}_k$ is the channel between the RIS and user $k$, and ${\bf\Theta}$ is the diagonal matrix, encompassing the RIS coefficients. 
Assuming a nearly passive diagonal RIS, the feasible set of the RIS coefficients is convex as
\begin{equation}
    \mathcal{D}=\{\theta_{mn}: |\theta_{mm}|^2 \leq 1,\theta_{mn} = 0,\forall m\neq n \},
\end{equation}
where $\theta_{mn}$ denotes the entry of the $m$-th row and $n$-th column of ${\bf\Theta}$. Note that we also consider non-convex feasibility sets for ${\bf\Theta}$ in Section \ref{sec-v-d}.

In this section, we also provide a few numerical examples using the simulation parameters and setup in Appendix \ref{app-sim}. Note that we are not aware of any other treatises solving the OPs in Section \ref{sec=iii} and Section \ref{sec=iii-max} for MU-MIMO systems using FBL coding. Hence, in the numerical results, we consider the following schemes:
\begin{itemize}
    \item \textbf{RIS} (or \textbf{RIS}$_i$): The solution of our framework for RIS-aided systems with feasibility set $\mathcal{D}$ (or $\mathcal{D}_i$).

    \item \textbf{No-RIS}: The solution of our framework for systems, operating without RIS.

    \item \textbf{RIS-Rand}: The solution of our framework for RIS-aided systems, employing random RIS elements.
\end{itemize}

 %\vspace{-.5cm}
\subsection{Sum Delay Minimization}\label{sec-v-c}
 {The optimization of $\{{\bf W} \}$ for fixed ${\bf\Theta}$ is provided in Section \ref{sec-sum-delay}, when the sum delay is minimized. Indeed, solving \eqref{eq-sdm-sur} gives $\{{\bf W}^{(z)} \}$. Here, we solve the sum delay minimization for fixed $\{{\bf W} \}$, which can be written as 
%\begin{subequations}
\begin{align}\label{eq-sdm-ris}
        \min_{{\bf \Theta}\in\mathcal{D} }\, & \sum_k\frac{L_k}{r_k}.
\end{align}
%\end{subequations}
To solve \eqref{eq-sdm-ris}, we employ an approach similar to the one in Section \ref{sec-sum-delay}. More specifically, we leverage Lemma 
\ref{cor-min} and a concave lower bound for $r_k$ as $\hat{r}_k$, provided in \cite[Corollary 1]{soleymani2024optimization}, which yields
%\begin{subequations}
\begin{align}\label{eq-sdm-ris-sur}
        \min_{{\bf \Theta}\in\mathcal{D},{\bf t} }\, & \sum_k\frac{L_k}{t_k}
        &
        \text{s.t.}\,\,&  \hat{r}_k\geq t_k>0,\,\, \forall k,  
        %\\        &&&|\theta_m|\leq 1, \,\,\, \forall m.
    \end{align}
%\end{subequations}
whose solution is ${\bf \Theta}^{(z)}$.} The algorithm converges to an SP of the sum delay minimization OP, since the algorithm falls into the MM framework.

%%\vspace{-.34cm}
%\vspace{-.5cm}
\subsection{Spectral-energy-efficiency Tradeoff}\label{sec-v-a}
%%\vspace{-.1cm}
 {The optimization of $\{{\bf W} \}$ for fixed ${\bf\Theta}$ is provided in Section \ref{sec-iv-a}, when the SEE tradeoff metrics are maximized. More specifically, $\{{\bf W}^{(z)}\}$ is obtained by solving \eqref{eq-sdm-sur}. Here, we provide a solution for updating ${\bf\Theta}$, when maximizing the tradeoff metric $\zeta$ stated in \eqref{eq-11}, yielding the following OP 
%\begin{subequations}
    \begin{align}\label{opt-13-see-theta}
        \max_{{\bf \Theta}\in\mathcal{D}}\,\, & \alpha\sum_k r_k
        +
        \frac{(1-\alpha)\sum_k r_k}{\sum_k[P_s+\eta \text{Tr}({\bf W}_k^{(z)} {\bf W}_k^{(z)^H})]}. 
    \end{align}
%\end{subequations}
To solve \eqref{opt-13-see-theta}, we employ the concave lower bound $\hat{r}_k$, leading to the following convex OP
%\begin{subequations}
    \begin{align}\label{opt-13-see-theta2}
        \max_{{\bf \Theta}\in\mathcal{D}}\,\, & \alpha\sum_k \hat{r}_k
        +
        \frac{(1-\alpha)\sum_k \hat{r}_k}{\sum_k[P_s+\eta \text{Tr}({\bf W}_k^{(z)} {\bf W}_k^{(z)^H})] }.
    \end{align}
%\end{subequations}
The solution of \eqref{opt-13-see-theta2} yields ${\bf\Theta}^{(z)} $, and the overall algorithm converges to an SP.}

\begin{figure}[t]
    \centering
    \begin{subfigure}[t]{0.24\textwidth}
        \centering
           \includegraphics[width=\textwidth]{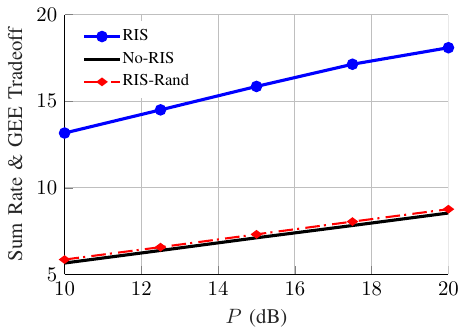}%{fig/sgee.pdf}
        \caption{Using the metric in \eqref{eq-11} with $\alpha=\frac{1}{2}$.}
    \end{subfigure}
%~~
\begin{subfigure}[t]{0.24\textwidth}
        \centering       \includegraphics[width=\textwidth]{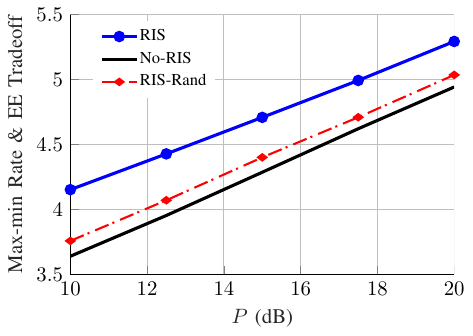}%{fig/see.pdf}
        \caption{Using the metric in \eqref{eq-see} with $\alpha_k=\frac{1}{2}$ for all $k$.}
    \end{subfigure}%
    \caption{SEE tradeoff versus $P$ for $N_{BS}=N_u=4$, and $K=3$.} %\vspace{-.5cm}
	\label{Fig-rr5} 
\end{figure} 

In Fig. \ref{Fig-rr5}, we show the average SEE tradeoff, using the metrics in \eqref{eq-11} and \eqref{eq-see}, versus the BS power budget. The RIS improves both the SEE tradeoff metrics. Moreover, the benefits of RIS are more significant, when sum rate and GEE are considered as the metrics for the SE and EE, respectively.  
%%\vspace{-.45cm}
\subsection{Sum EE Maximization}\label{sec-v-b}
%%\vspace{-.1cm}
 {The optimization of $\{{\bf W} \}$ for fixed ${\bf\Theta}$ can be found in Section \ref{sec-sum-ee}, when the sum EE is maximized. Particularly, $\{{\bf W}^{(z)} \}$ is calculated by solving \eqref{eq-45-sumEE}.
Upon employing the concave lower bound $\hat{r}_k$, we can compute ${\bf\Theta}^{(z)}$ by solving
%\begin{subequations}\label{eq-sum-ee}
    \begin{align}\label{eq-sum-ee-theta}
        \max_{{\bf \Theta}\in\mathcal{D} } & \sum_k
        \frac{\alpha_k \hat{r}_k}{P_s+\eta \text{Tr}({\bf W}_k^{(z)} {\bf W}_k^{(z)^H}) },
    \end{align}
%\end{subequations}
converging to an SP of the sum EE maximization problem.}

\begin{figure}[t]
    \centering
      \includegraphics[width=.44\textwidth]{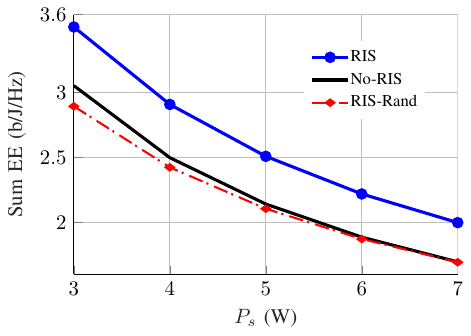}%{fig/sum-e.pdf}
    \caption{Average sum EE obtained by solving \eqref{eq-sum-ee} and \eqref{eq-sum-ee-theta} versus $P$ for $N_{{BS}}=N_{u}=5$, and $K=2$.} %%\vspace{-.7cm}
	\label{Fig-sum-e}  %\vspace{-.65cm}
\end{figure}
In Fig. \ref{Fig-sum-e}, we show the average sum EE versus $P_s$. In this example, the RIS improves the sum EE, when its elements are optimized based on our framework. However, the RIS may not provide any gain, if its elements are chosen randomly.

\subsection{ {Other Feasibility Sets for the RIS Coefficients} }\label{sec-v-d}

 {In this subsection, we detail how one can leverage the frameworks proposed in this paper to optimize non-convex sets of RIS coefficients. To this end, we present three additional feasibility sets for the scattering matrix of nearly passive RISs, based on the models in \cite[Sec. II.B]{wu2021intelligent}. Then, we propose a suboptimal approach to `convexify' these non-convex sets. Note that the other parts of our solutions, including calculating concave lower bounds for the rates and transforming a complicated FMP OP to a simpler surrogate OP, remain unchanged, when considering a non-convex set for the feasible values of ${\bf \Theta}$. To elaborate on this issue, we solve the surrogate OP in \eqref{eq-sdm-ris} with these non-convex sets. Additionally, note that the algorithms converge, since they generate a monotone sequence of the OFs. However, we do not make any claim on the optimality of the resulting solutions, when the feasibility set of ${\bf \Theta}$ is non-convex.}

 {Firstly, we consider the following set satisfying the unit modulus constraint:
\begin{equation}
   \mathcal{D}_{1}=\left\{\theta_{mn}:|\theta_{mm}|= 1, \theta_{mn}=0,\forall m\neq n\right\}.
\end{equation}
The constraint $|\theta_{mm}|=1$ can be written as the two constraints $|\theta_{mm}|^2\leq 1$ and $|\theta_{mm}|^2\geq 1$. The former is convex, but the latter is not. To convexify constraint $|\theta_{mn}|^2\geq 1$, we relax the constraint and employ the lower bound in Lemma \ref{lem=5}, yielding 
\begin{equation}\label{+=+}
|\theta_{mm}|^2\!\geq\!|\theta_{mm}^{(z-1)}|^2\!-2\mathfrak{R}\{\theta_{mm}^{(z-1)^*}(\theta_{mm}-\theta_{mm}^{(z-1)})\}\!\geq 1-\varsigma,\!\!
\end{equation}
where $\varsigma>0$. Upon using \eqref{+=+}, the resultant surrogate OP for $\mathcal{D}_{1}$ is
\begin{subequations}\label{eq-sdm-ris-sur-t1}
    \begin{align}
        \min_{{\bf \Theta},{\bf t} }\, & \sum_k\frac{L_k}{t_k}
        &
        \text{s.t.}\,\,& |\theta_{mm}|^2\leq 1, \, \text{and} \,\, \eqref{+=+},\,\,\, \forall m, 
        \\        &&&  \hat{r}_k\geq t_k>0,\,\, \forall k, \label{eq-62b}
    \end{align}
\end{subequations}
whose solution is denoted by $\hat{{\bf \Theta}}=\text{diag}(\hat{\theta}_{11},\hat{\theta}_{22},\cdots,\hat{\theta}_{MM})$, where $M$ is the number of RIS elements, and $\text{diag}(\cdot)$ denotes a diagonal matrix. Due to the relaxation in \eqref{+=+}, $\hat{{\bf \Theta}}$ may not be in $\mathcal{D}_{1}$. To address this issue, we project  $\hat{{\bf \Theta}}$ to  $\mathcal{D}_1$ by normalizing its diagonal elements as $\hat{\theta}_{mm}^{\text{new}}=\frac{\hat{\theta}_{mm}}{|\hat{\theta}_{mm}|}$ for all $m$. To ensure convergence, we update ${\bf \Theta}$ as
\begin{equation}\label{eq-42}
\bm{\Theta}^{(z)}=
\left\{
\begin{array}{lcl}
\hat{\bm{\Theta}}^{\text{new}}&\text{if}&
\sum_kd_k\left(\left\{\mathbf{W}^{(z)}\right\},\hat{{\bf\Theta}}^{\text{new}}\right)\leq
\\
&&
\sum_kd_k\left(\left\{\mathbf{W}^{(z)}\right\},{\bf \Theta}^{(z-1)}\right)
\\
\bm{\Theta}^{(z-1)}&&\text{Otherwise},
\end{array}
\right.
\end{equation}
where $\hat{\bm{\Theta}}^{\text{new}}=\text{diag}(\hat{\theta}_{11}^{\text{new}},\hat{\theta}_{22}^{\text{new}},\cdots,\hat{\theta}_{MM}^{\text{new}})$.
The updating rule in \eqref{eq-42} guarantees convergence because a sequence of non-increasing $\sum_kd_k$ is generated. 
}

 {Another set considered is based on the model in \cite{abeywickrama2020intelligent}, assuming that the amplitude of each RIS coefficient is a deterministic function of its phase as
\begin{equation}\label{eq*=*}
\mathcal{F}(\angle \theta_{mi})= |\theta|_{\min}+( 1\!-|\theta|_{\min})\left(\!\!\frac{\sin\left(\angle \theta_{mi}-\varphi\right)+1}{2}\!\right)^{\varpi}\!\!\!,\!\!
\end{equation}
where $|\theta|_{\min}$, $\varpi$, and $\varphi$ are non-negative constant values. Additionally, $\angle x$ takes the phase of $x$.
This set can be written as  \cite[Eq. (7)]{soleymani2022noma}
\begin{multline}
\mathcal{D}_{2}=\left\{\theta_{mn}\!:|\theta_{mm}|= \mathcal{F}(\angle \theta_{mm}), \,\angle \theta_{mm}\in[-\pi,\pi],
\right. \\ \left.
\theta_{mn}=0,\forall m\neq n\right\}.
\end{multline}
To `convexify' $\mathcal{D}_{2}$, we first relax the dependency of the amplitude of $\theta_{mm}$ on its phase, resulting in
\begin{align}\label{eq=9-2}
 |\theta|_{\min}^2\leq |\theta_{mm}|^2&\leq 1,\,\,\,\,\forall m,
\end{align}
where $|\theta_{mm}|^2\leq 1$ is a convex constraint, but $|\theta|_{\min}^2\leq |\theta_{mm}|^2$ is non-convex. 
To tackle this issue, we derive a concave lower bound for $|\theta_{mn}|^2$ using Lemma \ref{lem=5}, yielding the convex constraint
\begin{equation}\label{eq-50-2}
|\theta_{mm}^{(z-1)}|^2+2\mathfrak{R}\left(\theta_{mm}^{(z-1)}(\theta_{mm}-\theta_{mm}^{(z-1)})^*\right)\geq |\theta|_{\min}^2.
\end{equation}
 Finally, the surrogate OP for $\mathcal{D}_{2}$ is
    \begin{align}
        \min_{{\bf \Theta},{\bf t} }\, & \sum_k\frac{L_k}{t_k}
        &
        \text{s.t.}\,\,& \eqref{eq-62b},\,\,  \eqref{eq-50-2},\,\, \text{and} \,\,|\theta_{mm}|^2\leq 1,  \, \forall m, 
    \end{align}
with the solution $\hat{{\bf \Theta}}=\text{diag}(\hat{\theta}_{11},\hat{\theta}_{22},\cdots,\hat{\theta}_{MM})$. To ensure a feasible solution, we project $\hat{{\bf \Theta}}$ into $\mathcal{D}_{2}$ by choosing $\hat{{\bf \Theta}}^{\text{new}}=\text{diag}(\hat{\theta}^{\text{new}}_{11},\hat{\theta}^{\text{new}}_{22},\cdots,\hat{\theta}^{\text{new}}_{MM})$, where $\hat{{\theta}}^{\text{new}}_{mm}=\mathcal{F}(\angle\hat{{\theta}}_{mm})$ for all $m$, where $\mathcal{F}$ is defined as in \eqref{eq*=*}. Finally, ${\bf\Theta}$ is updated based on \eqref{eq-42} to ensure convergence. 
}

 {The other feasibility set considered assumes a discrete domain for the phase of the RIS coefficients, while the amplitude is kept fixed at $1$ as 
\begin{multline}
\mathcal{D}_{3}=\left\{\theta_{mn}:|\theta_{mm}|= 1,\angle\theta_{mm}=\{\phi_1,\phi_2,\cdots,\phi_I\},
\right. \\ \left.
\theta_{mn}=0,\forall m\neq n\right\},
\end{multline}
where the $\phi_i$s for all $i$ are the only possible phase shifts that can be tuned \cite{wu2021intelligent}. 
To obtain a suboptimal solution for $\mathcal{D}_{3}$, we first relax the assumption of having a discrete domain for the phase of $\theta_{mm}$. By this relaxation, ${\bf\Theta}$ can take any value in $\mathcal{D}_{1}$. Then, we compute $\hat{{\bf \Theta}}$ by solving  \eqref{eq-sdm-ris-sur-t1}. Additionally, we project $\hat{{\bf \Theta}}$ into $\mathcal{D}_{3}$ by firstly normalizing its amplitude and then taking the closest phase to $\angle\hat{\theta}_{mm}$ in $\mathcal{D}_{3}$. The projection of $\hat{{\bf \Theta}}$ in $\mathcal{D}_{3}$ is denoted by $\hat{{\bf\Theta}}^{\text{new}}$. Finally, ${\bf \Theta}^{(z)}$ is calculated according to the rule in \eqref{eq-42} to guarantee convergence.}

\begin{figure}[t]
    \centering
      \includegraphics[width=.48\textwidth]{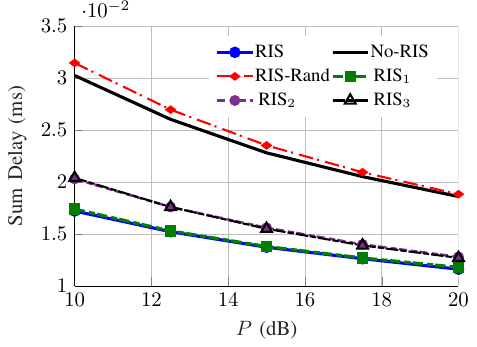}%{fig/sum-d.pdf}
    \caption{Average sum delay, using the metric in \eqref{eq-delay-18}, versus $P$ for $N_{{BS}}=4$, $N_{u}=3$, and $K=2$.} %\vspace{-.5cm}
	\label{Fig-sum-d} 
\end{figure}
 {In Fig. \ref{Fig-sum-d}, we show the average sum delay versus $P$, using Monte Carlo simulations. The sum delay of the algorithms considered in Fig. \ref{Fig-sum-d} decreases with the power budget of the BS. Moreover, for each feasibility set, the RIS substantially reduces the delay, when appropriately optimized. Interestingly, to achieve a target sum delay, the RIS-aided system requires a significantly lower transmission power, which indicates how energy-efficient RIS can be, when low-latency communication is required. Note that $I=8$ for $\mathcal{D}_{3}$ in this example, meaning that the phase of each RIS coefficient in $\mathcal{D}_{3}$ can only take values $\frac{i\pi}{4}$ for $i=1,2,\cdots,8$.}

\section{Comparison with Dinkelbach-based Algorithms} \label{sec-dink}

Our framework can also solve multiple-ratio max-min FMP problems, which are a special case of the FMP problems formulated in \eqref{(1)}. As a practical application of these OPs, we consider the maximization of the minimum EE, referred to as the max-min EE problem, as 
\begin{align}\label{eq-maxmin-ee}
    \max_{\{{\bf W} \} }\,&\min_k \frac{r_{k}}{P_s+\eta \text{Tr}({\bf W}_{k}{\bf W}_{k}^H ) }
    & \text{s.t.}\,\,&\eqref{eq-13-b},\,\,r_k\!\geq\! r_k^{th}\!,\forall k.\!\!
\end{align}
To find an SP of \eqref{eq-maxmin-ee} with the GDA, we first employ the bound in Lemma \ref{lem-1} and define the surrogate OP 
\begin{align}\label{eq-maxmin-ee2}
    \max_{\{{\bf W} \} }\,&\min_k \frac{\tilde{r}_{k}}{P_s+\eta \text{Tr}({\bf W}_{k}{\bf W}_{k}^H ) }
    & \text{s.t.}\,\,&\eqref{eq-13-b},\,\,\tilde{r}_k\!\geq\! r_k^{th}\!,\forall k.\!\!
\end{align}
Now we apply the GDA to \eqref{eq-maxmin-ee2}, which means that we iteratively solve
\begin{align}\label{eq-maxmin-ee3}
    \max_{\{{\bf W} \} }\,&\min_k\left\{ {\tilde{r}_{k}}-\mu^{(m)}p_k
    \right \}
    & \text{s.t.}\,\,&\eqref{eq-13-b},\,\,\tilde{r}_k\!\geq\! r_k^{th}\!,\forall k.\!\!
\end{align}
and update $\mu^{(m)}$ as
\begin{equation}\label{mu}
    \mu^{(m)}=\min_k \frac{\tilde{r}_{k} (\{{\bf W}^{(m-1)} \})}{P_s+\eta \text{Tr}({\bf W}_{k}^{(m-1)}{\bf W}_{k}^{(m-1)^H} ) },
\end{equation}
where $\{{\bf W}^{(m-1)} \}$ is the solution of \eqref{eq-maxmin-ee3} in the $(m-1)$-st iteration of the GDA (i.e., the inner loop). 
Therefore, to solve \eqref{eq-maxmin-ee} with the GDA, we have to employ a twin-loop algorithm in which the outer loop is based on MM, while the inner loop is based on the GDA. The GDA solution for \eqref{eq-maxmin-ee} is summarized in Algorithm III. 
\doublespacing %%\vspace{-.2cm}
\begin{table}[htb]
%\scriptsize
\small
%\label{alg-ii}
\begin{tabular}{l}
\hline 
 \textbf{Algorithm III}: The two-loop GDA algorithm to solve \eqref{eq-maxmin-ee}.  \\
\hline
\hspace{0.2cm}{\textbf{Initialization}}\\
\hspace{0.2cm}Set $\delta_1$, $\delta_2$, $z=1$,  $\{\mathbf{W}\}=\{\mathbf{W}^{(0)}\}$\\
\hline %\hline
\hspace{0.2cm}
\textbf{While} $\left(\underset{\forall k}{\min}\,e^{(z)}_{k}-\underset{\forall k}{\min}\,e^{(z-1)}_{k}\right)/\underset{\forall k}{\min}\,e^{(z-1)}_{k}\geq\delta_1$\\ 
\hspace{.6cm}{Calculate $\tilde{r}_{lk}$, 
using Lemma \ref{lem-1}}\\ 
\hspace{.6cm}\textbf{While} $\left(\underset{\forall k}{\min}\,e^{(m)}_{k}-\underset{\forall k}{\min}\,e^{(m-1)}_{k}\right)/\underset{\forall k}{\min}\,e^{(m-1)}_{k}\geq\delta_2$\\ 
\hspace{1.2cm}{Calculate $\mu^{(m)}$ using \eqref{mu}}\\
\hspace{1.2cm}{Calculate $\{{\bf W}\}$ by solving \eqref{eq-maxmin-ee3}}\\
\hspace{1.2cm} $m=m+1$\\
\hspace{.6cm}\textbf{End (While)}\\
\hspace{.6cm}$z=z+1$\\
\hspace{0.2cm}\textbf{End (While)}\\
\hspace{0.2cm}{{\bf Return} $\{\mathbf{W}^{(\star)}\}$}\\
\hline 
\end{tabular} %%\vspace{-.6cm}
\end{table}
\singlespacing

In contrast to the Dinkelbach-based algorithms, our framework requires only a single-loop, which may reduce the implementation and computational complexities. Specifically, upon utilizing our framework, an SP of \eqref{eq-maxmin-ee} can be found by iteratively solving
\begin{subequations}\label{eq-our-sol}
    \begin{align}
        \max_{\{{\bf W}\},{\bf t} }\, & \min_k\left\{
        2\beta_k^{(z)} t_k-\beta_k^{(z)^2}{(P_s+\eta \text{Tr}({\bf W}_k{\bf W}_k^H)) }
        \right\}
        \\
        \text{s.t.}\,\,& \eqref{eq-13-b},\,\,
        \alpha_k \tilde{r}_k-t_k^2\geq 0,\,\,\tilde{r}_k\geq r_k^{th}, \,\,\, \forall k.
    \end{align}
\end{subequations}
We summarize our solution for \eqref{eq-maxmin-ee} in Algorithm IV. 
\doublespacing %%\vspace{-.2cm}
\begin{table}[htb]
%\scriptsize
\small
%\label{alg-ii}
\begin{tabular}{l}
\hline 
 \textbf{Algorithm IV}: Our single-loop framework to solve \eqref{eq-maxmin-ee}.  \\
\hline
\hspace{0.2cm}{\textbf{Initialization}}\\
\hspace{0.2cm}Set $\delta$,  $z=1$,  $\{\mathbf{W}\}=\{\mathbf{W}^{(0)}\}$\\
\hline %\hline
\hspace{0.2cm}
\textbf{While} $\left(\underset{\forall k}{\min}\,e^{(z)}_{k}-\underset{\forall k}{\min}\,e^{(z-1)}_{k}\right)/\underset{\forall k}{\min}\,e^{(z-1)}_{k}\geq\delta$\\ 
\hspace{.6cm}{Calculate $\tilde{r}_{lk}$, 
using Lemma \ref{lem-1}}\\ 
\hspace{.6cm}{Calculate $\{{\bf W}\}$ by solving \eqref{eq-our-sol}}\\
\hspace{.6cm}$z=z+1$\\
\hspace{0.2cm}\textbf{End (While)}\\
\hspace{0.2cm}{{\bf Return} $\{\mathbf{W}^{(\star)}\}$}\\
\hline 
\end{tabular}  %%\vspace{-.5cm}
\end{table}
\singlespacing
%%\vspace{-.4cm}
%\vspace{-.5cm}
\subsection{Performance Comparison}
The GDA and our framework yield the same type of solution, i.e., convergence to an SP of \eqref{eq-maxmin-ee}. In Fig. \ref{Fig-nfp-gda}, we compare the performance of our framework to that of Dinkelbach-based algorithms via Monte Carlo simulations. The simulation parameters and setup are described in Appendix \ref{app-sim}.
\begin{figure}[t]
    \centering
    \begin{subfigure}[t]{0.24\textwidth}%{0.42\textwidth}
        \centering
           \includegraphics[width=\textwidth]{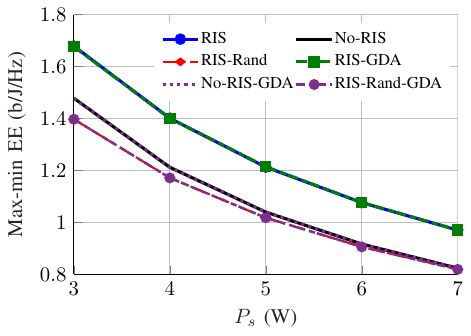}%{fig/ee.pdf}
        \caption{Max-min EE comparison by solving \eqref{eq-maxmin-ee}.}
    \end{subfigure}
%~~
\begin{subfigure}[t]{0.24\textwidth}%{0.42\textwidth}
        \centering       \includegraphics[width=\textwidth]{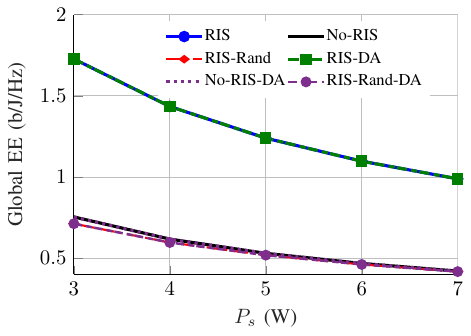}%{fig/gee.pdf}
        \caption{Comparison of GEE maximization, using the metric in \eqref{eq-gee}.}
    \end{subfigure}%
    \caption{Comparison of our framework and Dinkelbach-based algorithms for $P=10$ dB, $N_{BS}=N_u=5$, and $K=2$.} %%\vspace{-.5cm}
	\label{Fig-nfp-gda} 
\end{figure}
In these examples, our framework performs very similarly to the Dinkelbach-based algorithms. However, the Dinkelbach-based algorithms cannot solve any of the OPs considered in Section \ref{sec=iii},  and Section \ref{sec=iii-max}. Indeed, our framework is more general than the Dinkelbach-based algorithms and can solve a large variety of FMP problems. Moreover, the Dinkelbach-based algorithms require a twin-loop implementation, which is much more difficult to implement than our framework.

 %\vspace{-.5cm}
\subsection{Computational Complexity Comparison}
We approximate the number of multiplications to obtain a solution of \eqref{eq-maxmin-ee}  to compare the computational complexities of our framework and that of the GDA. These algorithms iteratively update the parameters by solving a convex OP. According to \cite[Chapter 11]{boyd2004convex}, the number of Newton iterations required to solve a convex OP increases with the number of constraints. The GDA updates $\{{\bf W}\}$ by solving \eqref{eq-maxmin-ee3}, having $2K+1$ constraints. To solve each Newton iteration, $K$ surrogate functions have to be calculated for the rates, which approximately requires $KN_{BS}^2(2N_{BS}+N_{u})$ multiplications. Hence, the approximate number of multiplications to solve \eqref{eq-maxmin-ee3} is $\mathcal{O}\left[N_{BS}^2K\sqrt{2K+1}(2N_{BS}+N_{u})\right]$.  {Similarly, the number of multiplications required to update $\{{\bf W}\}$ in each iteration of our framework can be approximated as $\mathcal{O}\left[N_{BS}^2K\sqrt{3K+1}(2N_{BS}+N_{u})\right]$, which is on the same order as the number of multiplications required for the GDA. As the GDA involves two loops, the total number of iterations necessitated for solving \eqref{eq-maxmin-ee} by the GDA can be limited to $MN$, where $M$ and $N$ are the maximum numbers of iterations for the inner and outer loops, respectively. However, our framework involves only a single loop, which makes its implementation easier than that of the GDA and may reduce the required total number of iterations.} 

 {To sum up, the computational complexity of updating $\{{\bf W}\}$ in each step is of the same order for both algorithms. More specifically, the approximate number of multiplications increases with $K^{3/2}$, $N_{BS}^3$, and $N_u$ for these algorithms. However, as the GDA is a dual-loop algorithm, the total number of updates for the GDA might be higher than for our framework.}
\begin{figure}[t]
    \centering
    \begin{subfigure}[t]{0.24\textwidth}%{0.42\textwidth}
        \centering
           \includegraphics[width=\textwidth]{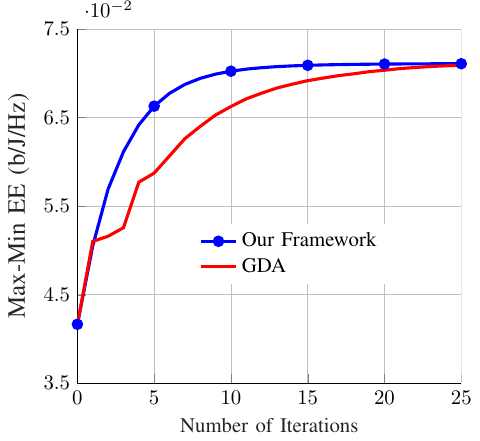}%{fig/ee-conv.pdf}
        \caption{Max-min EE comparison by solving \eqref{eq-maxmin-ee}.}
    \end{subfigure}
%~~
\begin{subfigure}[t]{0.24\textwidth}%{0.42\textwidth}
        \centering       \includegraphics[width=\textwidth]{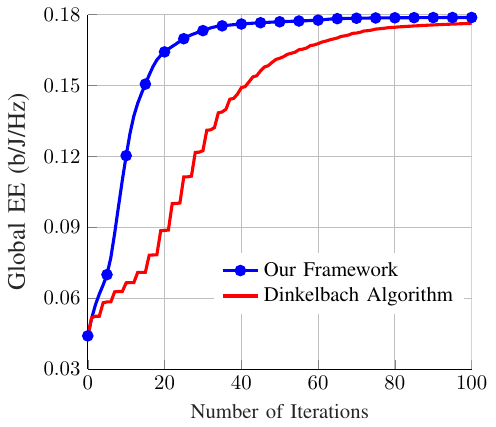}%{fig/gee-conv.pdf}
        \caption{Comparison of GEE maximization, using the metric in \eqref{eq-gee}.}
    \end{subfigure}%
    \caption{Convergence comparison of our framework and Dinkelbach-based algorithms for $P=10$ dB, $P_s=5$ W, $N_{BS}=N_u=2$, and $K=5$.} %\vspace{-.5cm}
	\label{Fig-nfp-gda-2} 
\end{figure}
In Fig. \ref{Fig-nfp-gda-2}, we show the convergence of our algorithm and Dinkelbach-based algorithms, utilizing Monte Carlo simulations. The number of iterations for the Dinkelbach-based algorithms equals the number of updating $\{{\bf W} \}$ to make a fair comparison. Our FMP solver converges faster than Dinkelbach-based algorithms, leading to less computational complexities.  {The average number of iterations until convergence to a max-min EE in Fig. \ref{Fig-nfp-gda-2}a  is $24$ for our framework and $33$ for the GDA. Moreover, our solver requires $76$ iterations on average to solve the maximization of GEE, while the Dinkelbach algorithm needs an average of $169$ iterations. Note that the convergence of our framework is smoother than that of Dinkelbach-based algorithms, since our solver is a single-loop algorithm. Indeed, the inner loop of Dinkelbach-based algorithms improves the OF much less than the outer loop. Hence, the steps for Dinkelbach-based algorithms seen in Fig. \ref{Fig-nfp-gda-2} correspond to updating the outer loop, making the curves non-smooth.}

%\vspace{-.3cm}
\section{Extension to Other Practical Scenarios}\label{sec-new-vii}
 Our general FMP solver is applicable to diverse MU-MIMO systems such as RIS-aided systems, multiple-access channels, multi-user interference channels (ICs), multi-cell BCs, and cognitive radio systems. In this work, we provide a practical example for single-cell BCs operating with and without RIS, assuming perfect devices and global channel state information (CSI). However, our framework is much more general than the system models and OPs described in Section \ref{sec=iii}, Section \ref{sec=iii-max}, and Section \ref{sec=iv}. In this section, we provide insights on how one can apply the general FMP solver to other practical scenarios, specifically in connection with the existing works in the literature. To this end, we do not focus on only MU-MIMO FBL systems.

\subsection{Applications in Systems with Hardware Impairments}
 In this subsection, we elaborate on how our framework applies to multiple-antenna MU systems affected by various sources of hardware imperfections. A common approach for studying hardware impairments (HWI) is to model the aggregated impact of these imperfections as additive Gaussian noise, having a variance proportional to the power of the transmitted/received signal \cite{xia2015hardware, soleymani2019improper, bjornson2013new, cheng2018performance, soleymani2020rate}. Here, we consider the scenario examined in one of these references, specifically \cite{soleymani2020rate}, which derives the rate region of a $K$-user MIMO IC by modeling HWIs as additive Gaussian noise and solving the maximization of the minimum weighted rate as formulated in \cite[Eq. (6)]{soleymani2020rate}. Instead of the minimum weighted rate OP, our frameworks can be used to solve FMP OPs such as those discussed in Section \ref{sec=iii} and Section \ref{sec=iii-max}. In particular, one can employ the concave lower bounds for the user rates developed in \cite[Theorem 1]{soleymani2020rate} and apply our solution in Section \ref{sec-sum-delay} to minimize the sum delay of users in the $K$-user MIMO IC with HWI.

Another common HWI is I/Q imbalance (IQI) \cite{boulogeorgos2016energy, javed2019multiple, soleymani2020improper}. When a device suffers from IQI, its output signal becomes a widely linear transformation of its input signal, rendering the resultant signal improper. In \cite{soleymani2020improper}, an optimization framework for the $K$-user MIMO IC with IQI was proposed to maximize the weighted sum rate, the weighted minimum rate/EE, and the GEE. By leveraging our FMP solver, the framework in \cite{soleymani2020improper} can be extended to tackle the FMP OPs presented in Section \ref{sec=iii} and Section \ref{sec=iii-max}. In particular, one can use \cite[Lemma 2]{soleymani2020improper} to transform the improper signals/noise caused by IQI into real-domain signals. Then, a concave lower bound for the rates can be derived based on the results in \cite[Lemma 5]{soleymani2020improper}. Once the lower bounds are calculated, Lemma \ref{cor-min} and Lemma \ref{lem-max-sur} can be used to solve the resulting minimization and maximization FMP problems, respectively.

\subsection{Applications in Systems with Imperfect CSI}
 {The practical examples considered in the previous sections assume perfect global CSI. However, the proposed framework can also be employed when developing robust designs following the worst-case robustness principle conceived for noisy CSI. Indeed, the results in Section \ref{secii-new} are general and independent of CSI assumptions. For robust designs following the worst-case robustness principle, each channel is assumed to belong to an uncertainty set with a specified probability \cite{vorobyov2003robust}. Then system parameters, including beamforming vectors or RIS scattering matrices, are designed to guarantee a specific worst-case performance for all channels within the uncertainty set \cite{vorobyov2003robust}. To deal with the uncertainty set, one can use the S-procedure or its generalization presented in \cite[Lemma 1]{yu2020robust}. Following this principle, the authors of \cite{li2024robust} proposed a robust rate-splitting scheme for maximizing the minimum rate in MISO BCs with FBL. Upon leveraging our FMP frameworks, the scheme in \cite{li2024robust} can be extended to solve the FMP OPs of Section \ref{sec=iii} and Section \ref{sec=iii-max} in MU-MISO FBL BCs.}

\subsection{Other Future Research Lines}
This subsection outlines additional future research directions for our FMP framework.  {In the practical examples considered in this paper, we focus on Gaussian signals, although the proposed framework itself is independent of the signaling schemes employed. An interesting research direction is to study OPs obtained for non-Gaussian signals and discrete constellations, which might be even more practical.} Moreover, we explain how the framework can be applied to systems with imperfect devices, considering two widely-used HWI models. Another timely research avenue is to investigate the impact of low-resolution analog-to-digital converters (ADCs), especially in applications to ultra-wide-band systems. Finally, we highlight three other promising applications of our FMP framework in the following.

\subsubsection{Emerging Technologies for Intelligent Surfaces}
In this paper, we study only a nearly passive diagonal RIS. However, the framework can be used in other RIS architectures, including globally passive \cite{fotock2023energy, soleymani2024rate2, soleymani2024energy}, non-diagonal RIS \cite{santamaria2023snr, soleymani2023optimization, li2022beyond, santamaria2024mimo}, active RIS \cite{zhang2022active}, multi-sector \cite{soleymani2024maximizing}, simultaneously transmitting and reflecting (STAR) RIS \cite{ahmed2023survey, liu2021star, soleymani2023spectral} and stacked intelligent surfaces \cite{an2023stacked}. Additionally, the framework is compatible with other emerging technologies such as fluid antennas.  

\subsubsection{Next Generations of Multiple Access}
We consider only treating interference as noise (TIN) in this paper. However, the framework is suitable for complex multiple-access (MA) schemes, utilizing successive interference cancellation (SIC). These MA schemes include, for instance, RSMA \cite{mao2022rate} and non-orthogonal MA (NOMA) \cite{ding2017survey}. To this end, the framework in \cite{soleymani2022rate} for RSMA and in \cite{soleymani2022noma} for NOMA-based schemes can be insightful. 

\subsubsection{Mobile-edge Computing} Latency minimization is a primary objective in applications related to mobile-edge computing. In these applications, the optimization variables include not only transmission parameters at the transmitters/RIS but also the number of bits offloaded to a computing center via a communication link \cite{bai2020latency, li2023min}. Our FMP framework in Section \ref{sec-ii-min} provides a powerful tool for developing latency-aware schemes for MU-MIMO RIS-aided mobile-edge computing systems operating under FBL coding.

\section{Summary and Conclusion}\label{sec-conclusion}
This paper proposed a unified framework for FMP problems, solving a pair of generic problems: a minimization and a maximization. FMP OPs are essential for wireless communication systems, since they include a wide range of practical OPs such as maximizing the SEE tradeoff metrics, weighted sum EE, the geometric mean of EE, minimizing sum delay, the geometric mean of delays, sum and maximum of MSE, to mention a few. In each generic OP, the OF and CFs are functions of any arbitrary and non-negative FFs. These generic problems typically cannot be tackled by traditional FP solvers, such as Dinkelbach-based algorithms. Indeed, Dinkelbach-based algorithms can only solve single- and multiple-ratio FMP OPs, but they are unable to solve, e.g., the minimization of the sum delay or the maximization of the sum EE.

We extensively compared the performance of our general framework to Dinkelbach-based algorithms. First, we showed that the proposed framework is more general than Dinkelbach-based algorithms, since it can solve OPs that Dinkelbach-based algorithms cannot. Second, our framework yields a single-loop iterative algorithm. By contrast, Dinkelbach-based algorithms require two loops for solving single- or multiple-ratio FMP OPs, such as max-min EE or maximizing the GEE. Thus, the implementation complexity of our framework is lower than that of its traditional counterparts.

In addition to developing the unified FMP solver, we proposed resource allocation schemes for MU-MIMO URLLC BCs as a practical example of the applications of our general optimization framework in wireless communication systems. We solved four practical OPs for each of the minimization and maximization solvers proposed in Section \ref{secii-new}. Specifically, we solved the minimization of sum delay, the geometric mean of delays, and the sum and maximum of MSE as practical examples of the generic minimization problems formulated in Section \ref{sec-ii-min}. Moreover, we solved the maximization of SEE tradeoff metrics, sum and geometric EE, and the weighted sum SINR by leveraging the framework in Section \ref{sec+ii}. We also showed how one can utilize our FMP solver in RIS-aided systems as another practical example of MU-MIMO systems. Additionally, we provided numerical results, showing that RIS improves the sum delay, SEE tradeoff and sum EE.

\appendices 
\section{Preliminaries on Majorization Minimization}\label{app-mm}
MM is an iterative method, converging to an SP of the original OP. The following lemma briefly states the main concept of MM. 
\begin{lemma}[\!\!\cite{aubry2018new}]\label{lem-mm}
Consider the non-convex problem
%\begin{subequations}
\begin{align}\label{app-ar-opt}
 \underset{\{\mathbf{X}\}\in\mathcal{X}
 }{\max}\,\,  & 
  \phi_0\left(\{\mathbf{X}\}\right) &
 \,\,\,\,\, \,\, \text{\em s.t.}   \,\,\,\,\,&  \phi_i\left(\{\mathbf{X}\}\right)\geq0,&\forall i,
 \end{align}
%\end{subequations}
where $\{\mathbf{X}\}$ is defined as in \eqref{(1)}, $\mathcal{X}$ is a convex set, and $\phi_i\left(\{\mathbf{X}\}\right)$ is a non-concave function.
An SP of \eqref{app-ar-opt} can be obtained by iteratively solving 
\begin{align}\label{app-ar-opt2}
 \underset{\{\mathbf{X}\}\in\mathcal{X}
 }{\max}\,\,  & 
  \tilde{\phi}_0^{(z)}\left(\{\mathbf{X}\}\right) &
 \,\,\,\,\, \,\, \text{\em s.t.}   \,\,\,\,\,&  \tilde{\phi}_i^{(z)}\left(\{\mathbf{X}\}\right)\geq0,&\forall i,
 \end{align}
where $z$ is the iteration index, and $\tilde{\phi}_i$ is a surrogate function that has to satisfy the conditions:
\begin{itemize}
\item $\tilde{\phi}_i^{(z)}\left(\{\mathbf{X}^{(z-1)}\}\right)=\phi_i\left(\{\mathbf{X}^{(z-1)}\}\right)$. 

\item $\frac{\partial \tilde{\phi}_i^{(z)}\left(\{\mathbf{X}^{(z-1)}\}\right)}{\partial \mathbf{X}_n}
=\frac{\partial \phi_{i}\left(\{\mathbf{X}^{(z-1)}\}\right)}{\partial \mathbf{X}_n}$, $\forall n$

\item $\tilde{\phi}_i^{(z)}(\{\mathbf{X}\})\leq \phi_i(\{\mathbf{X}\})$ for $\{\mathbf{X}\}\in\mathcal{X}$,
\end{itemize}
for all $i$,
where $\{\mathbf{X}^{(z-1)}\}$ is the solution of \eqref{app-ar-opt2} at the $(z-1)$-st iteration. 

\end{lemma}

Note that the function $\phi_i(\{{\bf X}\})$ and its surrogate function $\tilde{\phi}_i(\{{\bf X}\})$ have to be equal and have the same derivative at $\{{\bf X}^{(z)} \}$. Moreover, $\tilde{\phi}_i(\{{\bf X}\})$ has to majorize $\phi_i(\{{\bf X}\})$  for the entire domain, i.e., $\tilde{\phi}_i(\{{\bf X}\})\geq\phi_i(\{{\bf X}\})$ $\forall\{{\bf X}   \}\in \mathcal{X}$.

 %\vspace{-.2cm}
\section{Proof of Theorem \ref{th1-min}}\label{app-1-min}
Let us define the variables $u_{mi}\geq0$ and $t_{mi}>0$ for ensuring that $f_{mi}(\{{\bf X}   \})\leq u_{mi}^2$ and $g_{mi}(\{{\bf X}   \})\geq t_{mi}$, $\forall m,i$. 
In this case, we have $h_{mi}(\{{\bf X}   \})\leq \frac{u_{mi}^2}{t_{mi}}$, $\forall m,i$. Upon utilizing the auxiliary variables ${\bf t}=\{t_{mi}\}_{\forall m,i} $ and ${\bf u}=\{u_{mi}\}_{\forall m,i} $, we can rewrite \eqref{(1-min)} as
\begin{subequations}\label{(app-2-min)}
    \begin{align}
        \max_{\{{\bf X}   \}\in \mathcal{X}, {\bf t}, {\bf u}  } & \sum_{m=1}^{M_0}\frac{u_{m0}^{2}}{t_{m0}}
        &%\\
        \text{s.t.}&\,\,\eqref{eq-2b-min},\eqref{eq-2d-min},\eqref{eq-2e-min}, 
        \\&&
        \label{eq-2c-min-main}        &  u_{mi}^2-f_{mi}(\{{\bf X}   \})\!\geq\! 0, %2\sqrt{f_{mi}^{(z)}}u_{mi} - f_{mi}^{(z)}\geq f_{mi}(\{{\bf X}\}),
        \,\, \forall m,i,\!\!
    \end{align}
\end{subequations}
where the notations are defined as in Theorem \ref{th1-min}. 
 {The constraint in \eqref{eq-2c-min-main} is not convex, not even when $f_{mi}(\{{\bf X}   \})$ is convex in $\{{\bf X}   \}$, since $u_{mi}^2$ is convex, instead of being concave. Thus, we employ the concave lower bound in Lemma \ref{lem=5} to convexify \eqref{eq-2c-min-main} as
\begin{multline}
    \label{eq-78-n}
    u_{mi}^2-f_{mi}(\{{\bf X}   \})\geq \\
    2u_{mi}^{(z)}u_{mi}
    -\left(u_{mi}^{(z)}\right)^2-f_{mi}(\{{\bf X}   \})\geq 0,\,\forall m,i.
\end{multline}%\end{equation}
This approach is also known as the convex-concave procedure (CCP) %to  convexify \eqref{eq-2c-min-main} as in \eqref{eq-2c-min}, which proves the theorem because the CCP 
and falls into the MM framework \cite{sun2017majorization}. We set $u_{mi}^{(z)}=\sqrt{f_{mi}(\{{\bf X}^{(z)}   \})}$ for all $m,i$, since we assume $f_{mi}(\{{\bf X}   \})\leq u_{mi}^2$, and the bound in \eqref{eq-78-n} is tighter when the equality holds. Upon substituting \eqref{eq-2c-min-main} with \eqref{eq-78-n} and setting $u_{mi}^{(z)}=\sqrt{f_{mi}(\{{\bf X}^{(z)}   \})}$ for all $m,i$, the proof of Theorem \ref{th1-min} is complete.}

 %\vspace{-.5cm}
\section{Proof of Lemma \ref{cor-min}}\label{app-prof-lem1}
 {Similar to the proof of Theorem \ref{th1-min}, we first define  variables $u_{mi}\geq0$ and $t_{mi}>0$, satisfying $f_{mi}(\{{\bf X}   \})\leq u_{mi}^2$ and $g_{mi}(\{{\bf X}   \})\geq t_{mi}$, $\forall m,i$. Hence, \eqref{(1-min)} can be rewritten as
\begin{subequations}\label{(app-3-min)}
    \begin{align}
        \max_{\{{\bf X}   \}\in \mathcal{X}, {\bf t}, {\bf u}  } & \sum_{m=1}^{M_0}\frac{u_{m0}^{2}}{t_{m0}}
        &%\\
        \text{s.t.}&\,\, \eqref{eq-2b-min},\eqref{eq-2e-min},
        \\&&
        \label{eq-app-3b-min-main}        &  u_{mi}^2-f_{mi}(\{{\bf X}   \})\!\geq\! 0, %2\sqrt{f_{mi}^{(z)}}u_{mi} - f_{mi}^{(z)}\geq f_{mi}(\{{\bf X}\}),
        \,\, \forall m,i,\!\!
        \\&&&g_{mi}(\{{\bf X}\})\!-\!t_{mi}\geq 0, \, \forall m,i. \!\!
        \label{eq-app-3c-min-main}
    \end{align}
\end{subequations}
The OF of \eqref{(app-3-min)} is a convex function. Additionally, the constraints in \eqref{eq-2b-min} and \eqref{eq-2e-min}  are convex. Hence, when $f_{mi}(\{{\bf X}   \})$ and $g_{mi}(\{{\bf X}   \})$ are, respectively, non-convex and non-concave, the OP in \eqref{(app-3-min)} is non-convex due to \eqref{eq-app-3b-min-main} and \eqref{eq-app-3c-min-main}.
To ``convexify'' \eqref{eq-app-3b-min-main}, one can utilize the concave lower bound in Lemma \ref{lem=5} and the surrogate function $\tilde{f}_{mi}(\{{\bf X}   \})$, obeying the conditions stated in Lemma \ref{cor-min}, which results in
\begin{multline}\label{eq-80-n}
    u_{mi}^2-f_{mi}(\{{\bf X}   \})   \geq
    \\2u_{mi}^{(z)}u_{mi}-\left(u_{mi}^{(z)}\right)^2-\tilde{f}_{mi}(\{{\bf X}   \})\geq 0,\,\forall m,i.
\end{multline}
Note that we set $u_{mi}^{(z)}=\sqrt{f_{mi}(\{{\bf X}^{(z)}   \})}$ for all $m,i$, as explained in Appendix \ref{app-1-min}. Additionally, we can convexify \eqref{eq-app-3c-min-main} by leveraging the surrogate function $\tilde{g}_{mi}(\{{\bf X}   \})$, fulfilling the conditions stated in Lemma \ref{cor-min}. In this case, we can attain an SP of \eqref{(1-min)} since all the surrogate functions of \eqref{(2-min-cor)} satisfy the three conditions in Appendix \ref{app-mm}.
}
 %\vspace{-.5cm}
\section{Proof of Theorem \ref{th1}}\label{app-2}
Let us define the variables $t_{mi}\geq0$ so that $f_{mi}(\{{\bf X}   \})\geq t_{mi}^2$, $\forall m,i$. 
In this case, \eqref{(1)} can be equivalently rewritten as %follows  
\begin{subequations}\label{(1)-sur}
    \begin{align}
        \max_{\{{\bf X}   \}\in \mathcal{X}, {\bf t}  } & \sum_{m=1}^{M_0}\frac{t_{m0}^2}{g_{m0}(\{{\bf X}   \})}
        &%\\
        \text{s.t.}&\sum_{m=1}^{M_i}\frac{t_{mi}^2}{g_{mi}(\{{\bf X}   \})}\!\geq 0,\,  \forall i,\!\!\!
        \\ &&
        &\eqref{eq-3c},\eqref{eq-3d},
    \end{align}
\end{subequations}
which is still in a fractional format. 
Upon using \eqref{eq-43}, we have 
\begin{equation}
    \frac{t_{mi}^2}{g_{mi}(\{{\bf X}   \})}
     \geq 2a_{mi}^{(z)}t_{mi}- a_{mi}^{(z)^2}g_{mi}(\{{\bf X}\})\triangleq \tilde{h}_{mi}(\{{\bf X}   \}),
\end{equation}
where $a_{mi}^{(z)}=\frac{t_{mi}^{(z)}}{g_{mi}(\{{\bf X}^{(z)}\})}$, and $t_{mi}^{(z)}$ is a feasible point and can be chosen as %$t_{mi}^{(z)}=\sqrt{f_{mi}(\{{\bf X}^{(z)}\})}$. 
\begin{equation}
    t_{mi}^{(z)}=\sqrt{f_{mi}(\{{\bf X}^{(z)}\})}.
\end{equation}
Upon inserting the surrogate function $\tilde{h}_{mi}(\{{\bf X}   \})$ into \eqref{(1)-sur}, we have the surrogate OP in \eqref{(2)}. Since   $\tilde{h}_{mi}(\{{\bf X}   \})$ fulfills the three conditions in Appendix \ref{app-mm}, we can obtain an SP of \eqref{(1)-sur} by iteratively solving \eqref{(2)}. Moreover, the solution of \eqref{(1)-sur} is equal to the solution of \eqref{(1)}, which proves the theorem.

 %\vspace{-.4cm}
\section{Proof of Lemma \ref{lem-max-sur}}\label{app-prof-lem2} %%\vspace{-.1cm}
In the general case, \eqref{(2)} is non-convex since $f_{mi}(\{{\bf X}   \})$ and $g_{mi}(\{{\bf X}   \})$, $\forall m,i,$ are not concave and convex, respectively.
To address this issue, we can employ the surrogate functions $\tilde{f}_{mi}(\{{\bf X}   \})$ and $\tilde{g}_{mi}(\{{\bf X}   \})$ obeying the properties mentioned in Lemma \ref{lem-max-sur}.  In this case, the resultant surrogate OP in \eqref{(2-sur)} is convex. Furthermore, the algorithm in Lemma \ref{lem-max-sur} converges to an SP of \eqref{(1)}, since all the surrogate functions in \eqref{(2-sur)} fulfill the conditions of the MM framework stated in Appendix \ref{app-mm}. 
 %\vspace{-.5cm}
\section{Useful Inequalities}\label{app-1}
In this Appendix, we provide the inequalities employed throughout the paper. 
\begin{lemma}\label{1}
    Consider function $h({\bf X}, {\bf Y})$, which is jointly convex in arbitrary matrices ${\bf X}$ and ${\bf Y}$. Then, we can employ the first order Taylor expansion to obtain a linear lower bound for $h({\bf X}, {\bf Y})$ as follows \cite{sun2017majorization}
    \begin{multline}\label{(3)}
        h({\bf X}, {\bf Y})\geq h({\bf X}^{(t)} , {\bf Y}^{(t)})
        \\
        + \mathfrak{R} \left \{
        \text{\em Tr}
        \left (
        \left. \frac{\partial h({\bf X}, {\bf Y})}{{\partial \bf X}}\right|_{({\bf X}^{(z)}, {\bf Y}^{(z)})} \left({\bf X}-{\bf X}^{(z)} \right)^H
        \right )
        \right \}
        \\
        + \mathfrak{R} \left \{
        \text{\em Tr}
        \left (
        \left. \frac{\partial h({\bf X}, {\bf Y})}{\partial {\bf Y}}\right|_{({\bf X}^{(z)}, {\bf Y}^{(z)})} \left({\bf Y}-{\bf Y}^{(z)} \right)^H
        \right )
        \right \},
    \end{multline}
    where ${\bf X}^{(z)}$ and ${\bf Y}^{(z)}$ are any feasible points.
\end{lemma}
Note that the result in Lemma \ref{1} can be extended to the case with more than two variables. In case of scalar and real variables, \eqref{(3)} simplifies to 
\begin{multline}\label{neq-1}
        h(x,y)\geq h(x^{(t)} , y^{(t)})
        %\\
        +  \left. \frac{\partial h(x, y)}{{\partial x}}\right|_{(x^{(z)}, y^{(z)})} \left(x-x^{(z)} \right)
        \\
        +  \left. \frac{\partial h(x, y)}{{\partial y}}\right|_{(x^{(z)}, y^{(z)})} \left(y-y^{(z)} \right),
    \end{multline}
which can be used to derive useful inequalities such as 
\begin{align}\label{eq-43}
    \frac{x^2}{y}&\geq \frac{2\bar{x}}{\bar{y}}x-\frac{\bar{x}^2}{\bar{y}^2}y,
\end{align}
for real-valued $x$ and $y$.
Moreover, when $x$ is complex, we have
\begin{align}\label{eq-63}
    \frac{|x|^2}{y}&\geq \frac{2\mathfrak{R}\{\bar{x}^*x\}}{\bar{y}}-\frac{|\bar{x}|^2}{\bar{y}^2}y.
\end{align}

\begin{lemma}[\!\cite{soleymani2022improper}]\label{lem-2}
Consider arbitrary matrices ${\bf \Gamma}\in\mathbb{C}^{m\times n}$ and $\bar{\bf \Gamma}\in\mathbb{C}^{m\times n}$, and positive definite matrices ${\bf \Omega}\in\mathbb{C}^{m\times m}$ and $\bar{\bf \Omega}\in\mathbb{C}^{m\times m}$, where $m$ and $n$ are arbitrary natural numbers. Then, we have:
\begin{multline} 
\ln \left|\mathbf{I}+{\bf \Omega}^{-1}{\bf \Gamma}{\bf \Gamma}^H\right|\geq
 \ln \left|\mathbf{I}+{\bf \Omega}^{-1}\bar{{\bf \Gamma}}\bar{{\bf \Gamma}}^H\right|
\\-
\text{{\em Tr}}\left(
\bar{{\bf \Omega}}^{-1}
\bar{{\bf \Gamma}}\bar{{\bf \Gamma}}^H
\right)
+
2\mathfrak{R}\left\{\text{{\em Tr}}\left(
\bar{{\bf \Omega}}^{-1}
\bar{{\bf \Gamma}}{\bf \Gamma}^H
\right)\right\}\\
%&\hspace{.4cm}
-
\text{{\em Tr}}\left(
(\bar{{\bf \Omega}}^{-1}-(\bar{{\bf \Gamma}}\bar{{\bf \Gamma}}^H + \bar{{\bf \Omega}})^{-1})^H({\bf \Gamma}{\bf \Gamma}^H+{\bf \Omega})
\right).
\label{lower-bound}
\end{multline}
\end{lemma} 
\begin{lemma}[\cite{soleymani2024optimization}]\label{lem-3} 
 The following inequality holds for arbitrary matrices ${\bf \Gamma}\in\mathbb{C}^{m\times n}$ and $\bar{{\bf \Gamma}}\in\mathbb{C}^{m\times n}$ and positive semi-definite ${\bf \Omega}\in\mathbb{C}^{m\times m}$ and $\bar{{\bf \Omega}}\in\mathbb{C}^{m\times m}$
\begin{multline}%{equation}  
\label{eq10}
f\left({\bf \Gamma},{\bf \Omega}\right)=\text{\em Tr}\left({\bf \Omega}^{-1}{\bf \Gamma}{\bf \Gamma}^H\right)\geq 
2\mathfrak{R}\left\{\text{\em Tr}\left(\bar{\bf \Omega}^{-1}\bar{\bf \Gamma}{\bf \Gamma}^H\right)\right\}
\\-
\text{\em Tr}\left(\bar{\bf \Omega}^{-1}\bar{\bf \Gamma}\bar{\bf \Gamma}^H\bar{\bf \Omega}^{-1}{\bf \Omega}\right),
\end{multline}%{equation}
where $m$ and $n$ are arbitrary natural numbers. %, and $\mathfrak{R}\{x\}$ returns the real value of $x$.
\end{lemma}
\begin{lemma}\label{lem-gm}
The following inequality holds for real scalar variables $x_k\geq 0$ and $\bar{x}_k\geq 0$ for $k=1,\cdots,K$
\begin{equation}
    \prod_{k}x_k^{1/K} \leq \prod_{k}x_k^{1/K} + \sum_k \alpha_k (x_k-\bar{x}_k),
\end{equation}
where $\alpha_k=\frac{1}{K}\bar{x}_k^{\frac{1-K}{K}}\prod_{i\neq k}\bar{x}_i^{\frac{1}{K}}$, and $K>1$.
\end{lemma}
\begin{proof}
    Function $\prod_{k}x_k^{1/K}$ is jointly concave in ${\bf x}=\{x_1,x_2,\cdots,x_K\}$ for $x_k\geq 0$,  $\forall k$. Thus, we can employ the concave-convex procedure to obtain an upper bound for $\prod_{k}x_k^{1/K}$ \cite{lipp2016variations}.
\end{proof}

\begin{lemma}\label{lem-gm2}
For all real-valued $x_k\geq 0$ and $\bar{x}_k\geq 0$ for $k=1,\cdots,K$, we have
\begin{equation}
    \prod_{k}x_k^{2} \geq \prod_{k}\bar{x}_k^{2}+ \sum_k \alpha_k (x_k-\bar{x}_k),
\end{equation}
where $\alpha_k=2\bar{x}_k\prod_{i\neq k}\bar{x}_i^{2}$.
\end{lemma}
\begin{proof}
    The function $ \prod_{k}x_k^{2}$ is jointly convex in ${\bf x}=\{x_1,x_2,\cdots,x_K\}$ for $x_k\geq 0$, $\forall k$. Thus, we can employ the convex-concave procedure to obtain a lower bound for $\prod_{k}x_k^{1/K}$ \cite{lipp2016variations}.
\end{proof}%\vspace{-.5cm}
 {\begin{lemma}[\!\cite{soleymani2022rate}]\label{lem=5}
For all complex-valued $x_k$ and $\bar{x}_k$, we have
\begin{equation}
\sum_{k}|x_k|^2\geq \sum_{k}|\bar{x}_k|^2+2\sum_{k}\mathfrak{R}\left\{\bar{x}_k^{*}(x_k-\bar{x}_k)\right\}.
\end{equation}
\end{lemma}
}
 %%\vspace{-.6cm}
\section{Simulation Parameters}\label{app-sim}
\begin{figure}[t]
    \centering
      \includegraphics[width=.48\textwidth]{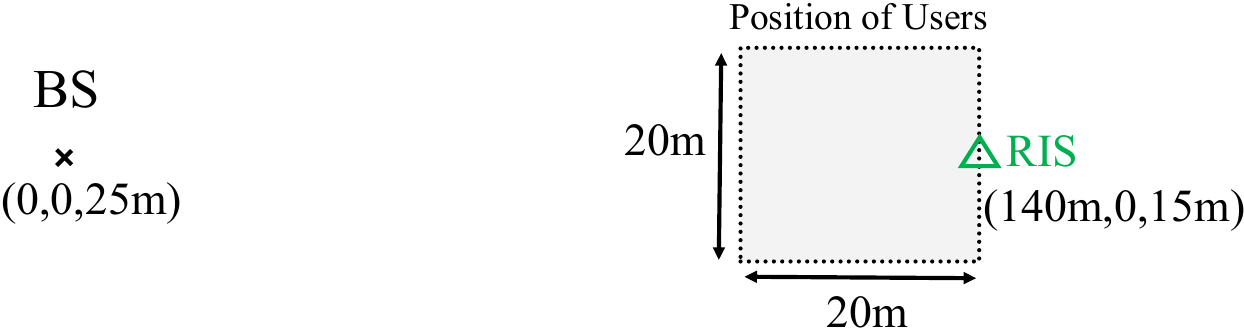}%{fig/sim-set.pdf}
    \caption{Simulation setup.} %\vspace{-.7cm}
	\label{Fig-sim-set} 
\end{figure}

 \begin{table*}
\centering
\scriptsize
\caption{ {List of selected simulation parameters.}}\label{table-sim}%%\vspace{-.4cm}
\begin{tabular}{ |c| c| c| c|c|c|c|c|c|c|c|c|c|c|c|c|c|}
\hline
 Parameter  & $\tau_{BS-RIS}$ & $\tau_{RIS-U}$ & $\tau_{BS-U}$ & PL$_{0,BS-RIS}$& PL$_{0,RIS-U}$& PL$_{0,BS-U}$&$\iota_{0,BS}$&$\iota_{0,RIS}$&Noise Power D. & Ch. Bw.& $N_{RIS} $ \\ 
\hline
Value & $2.2$ & $2.2$& $3.75$&$30$&$30$&$35.9$&$6$&$5.5$&$-174$ dB&$1.5$ MHz&20
 \\ 
\hline
\end{tabular}
\normalsize
%\vspace{-.6cm}
\end{table*}

In this appendix, we present the simulation parameters that are chosen based on \cite{soleymani2022improper, soleymani2024optimization}. More specifically, we assume that the direct channels between the BS and the users, $\tilde{\mathbf{G}}_{k}$ for all $k$, have a non-line of sight (NLoS) nature, thus following the Rayleigh distribution. It means that each entry of $\tilde{\mathbf{G}}_{k}^\prime$ for all $k$ is derived from a zero-mean complex proper Gaussian distribution with unit variance, where we have $\tilde{\mathbf{G}}_{k}=\rho_{k}\tilde{\mathbf{G}}_{k}^\prime $, and $\rho_{k}$ is the channel attenuation coefficient due to the large-scale path loss of the direct link between the BS and user $k$. Additionally, the channels between the BS and the RIS, ${\bf G}$, as well as the channels between the RIS and users, ${\bf G}_k$ for all $k$, have a line of sight (LoS) component, hence following the Rician distribution. In other words, these channels are given by \cite{pan2020multicell}
\begin{align}
    {\mathbf{G}}&=\sqrt{\frac{\varrho}{1+\varrho}}\tilde{\mathbf{G}}^{\text{LoS}}+\sqrt{\frac{1}{1+\varrho}}\tilde{\mathbf{G}}^{\text{NLoS}},\\
    {\mathbf{G}}_k&=\sqrt{\frac{\varrho}{1+\varrho}}\tilde{\mathbf{G}}^{\text{LoS}}_k+\sqrt{\frac{1}{1+\varrho}}\tilde{\mathbf{G}}^{\text{NLoS}}_k,
\end{align}
where $\varrho$ is the Rician factor, which is assumed to be $3$, $\tilde{\mathbf{G}}^{\text{LoS}}$ (or $\tilde{\mathbf{G}}^{\text{LoS}}_k$) is the line-of-sight (LoS) component of  channel ${\bf G}$ (or ${\bf G}_k$), and $\tilde{\mathbf{G}}^{\text{NLoS}}$ (or $\tilde{\mathbf{G}}^{\text{NLoS}}_k$) is the non-LoS component of  channel ${\bf G}$ (or ${\bf G}_k$), which is assumed to follow  Rayleigh fading similar to $\tilde{\mathbf{G}}_{k}^\prime$. 
The LoS component $\tilde{\mathbf{G}}^{\text{LoS}}$ is deterministic and given by $\tilde{\mathbf{G}}^{\text{LoS}}=\mathbf{a}_{N_r}\left(\psi^{A}\right)\mathbf{a}^H_{N_t}\left(\psi^{D}\right)$, 
where $\psi^{A}\sim\text{Unif}[0,2\pi]$ is the angle of arrival, $\psi^{D}\sim\text{Unif}[0,2\pi]$ is the angle of departure, $N_t/N_r$ is the number of transmit/receive antennas in the corresponding link, and we have
\begin{align} 
\mathbf{a}_{N_r}\!\left(\psi^{A}\right)&\!=\!\!\left[\!1,e^{j\frac{2\pi d\sin(\psi^{A})}{\lambda}},
%\cdots
...,e^{j\frac{2(N_r-1)\pi d\sin(\psi^{A})}{\lambda}}\right],\!
\\\!\!
\mathbf{a}_{N_t}\!\left(\psi^{D}
\right)&\!=\!\!\left[\!1,e^{j\frac{2\pi d\sin(\psi^{D})}{\lambda}},
%\cdots
...,e^{j\frac{2(N_t-1)\pi d\sin(\psi^{D})}{\lambda}}\right]\!\!,
\end{align}
where $d/\lambda$ is chosen as 1/2 for simplicity \cite{pan2020multicell}. 
Furthermore, the large-scale path loss in dB is given by \cite[Eq. (59)]{soleymani2022improper}
\begin{equation}
\text{PL}=\text{PL}_0+\iota_0-10\tau\log_{10}\left(\frac{\varphi}{\varphi_0}\right),
\end{equation}
where $\text{PL}_0$ is the path loss at the reference distance $\varphi_0=1$m, $\varphi$ is the link distance, $\tau$ is the path-loss exponent, and $\iota_0$ is the antenna gain at the transmitter side. Note that the channel attenuation coefficient is $\rho=10^{\text{PL}/10}$. 

The system topology is depicted in Fig. \ref{Fig-sim-set}. We assume that the BS is located at $(0,0,25)$m, where $25$m is the height of the BS. We further assume that the users having a height of $1.5$m are uniformly distributed in the $20$m$\times 20$m gray square area of Fig. \ref{Fig-sim-set}, which is centered at $(130,0,1.5)$m. We also assume that there is only a single RIS with height $15$m, located at $(140,0,15)$m. Furthermore, the simulation parameters are chosen based on \cite{soleymani2022improper}.  More specifically, the path-loss exponent is equal to $2.2$ and $3.75$ for the LoS and NLoS channels, respectively. Moreover, $\iota_0=6$ dB for the BS and  $\iota_0=5.5$ dB for the RIS. The path loss at the reference distance of $1$m is set to $-35.9$ dB for the direct link between the BS and each user and  $-30$ dB for the links related to the RIS. The channel bandwidth is $1.5$ MHz, and the noise power density is $-174$ dBm/Hz. Additionally, we set $N_{{RIS}}=20$, $\epsilon=10^{-5}$ and $n=256$. The simulation parameters are summarized in Table \ref{table-sim}.

\bibliographystyle{IEEEtran}

\bibliography{ref}

\begin{IEEEbiography}[{\includegraphics[width=1in,height=1.25in,clip,keepaspectratio]{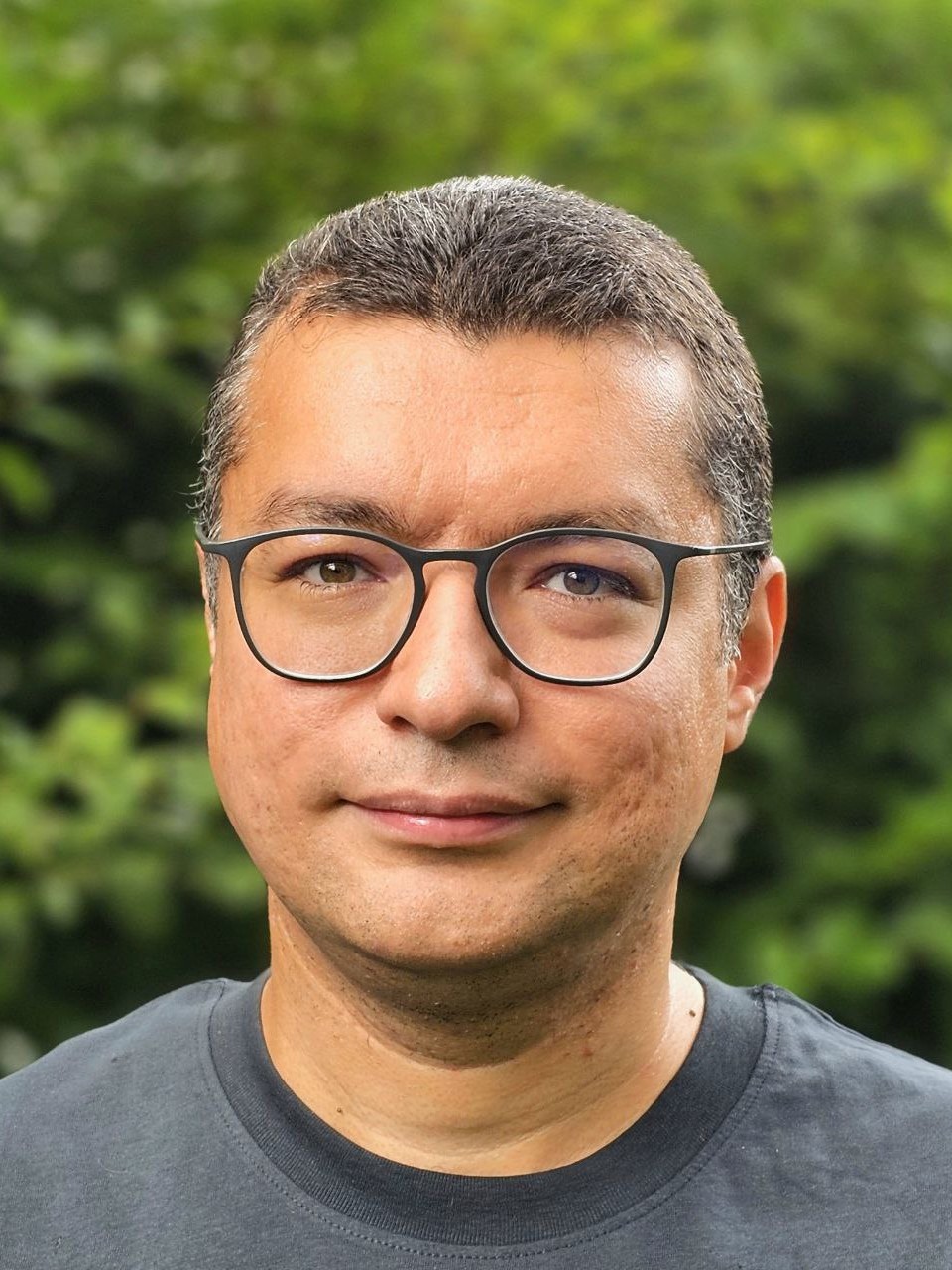}}]
{Mohammad Soleymani} was born in Arak, Iran. He received the B.Sc. degree from Amirkabir University of Technology (Tehran Polytechnic), the M.Sc. degree from Sharif University of Technology, Tehran, Iran, and the Ph.D. degree (with distinction) in electrical engineering from the University of Paderborn, Germany. He is currently an Akademischer Rat a. Z. with the Signal and System Theory Group at the University of Paderborn. He was a Visiting Researcher at the University of Cantabria, Spain. He serves on the editorial boards of ELSEVIER Signal Processing, EURASIP Journal on Wireless Communications and Networking, and Springer Journal of Wireless Personal Communications. His research interests include multi-user MIMO, wireless networking, numerical optimization, ultra-reliable low-latency communications (URLLC), and statistical signal processing.
\end{IEEEbiography}

\begin{IEEEbiography}[{\includegraphics[width=1in,height=1.25in,clip,keepaspectratio]{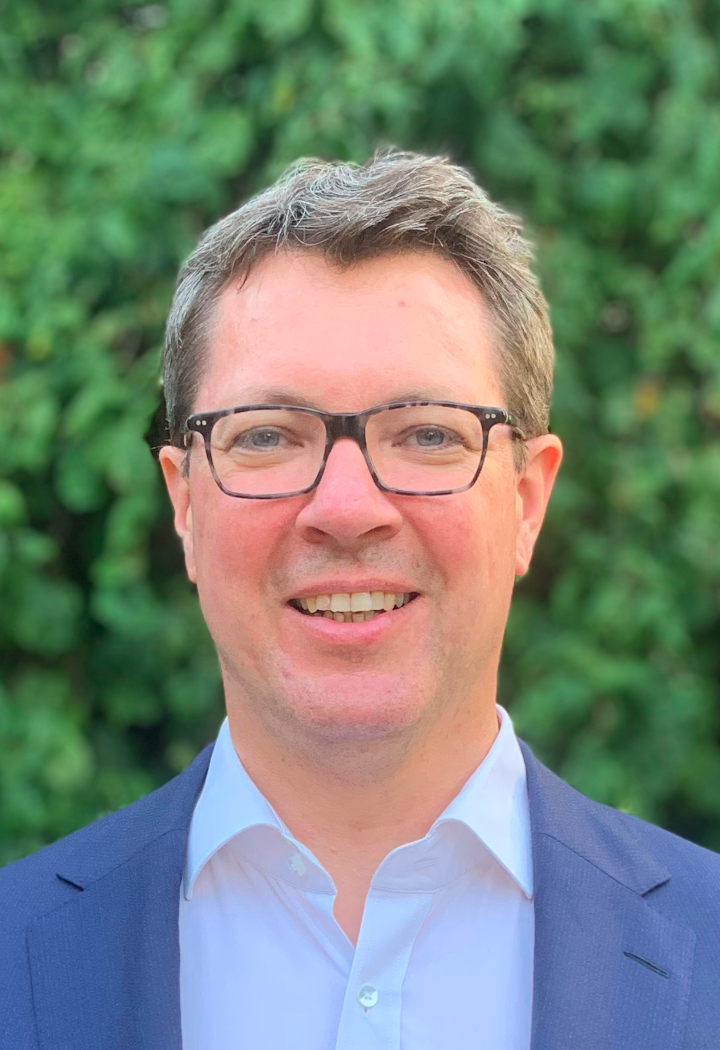}}]
{Eduard A. Jorswieck} (Fellow, IEEE) received the Ph.D. degree in electrical engineering and computer science from TU Berlin in 2004. From 2006 to 2008, he was with the Signal Processing Group, KTH Stockholm, as a Post-Doctoral Fellow and an Assistant Professor. From 2008 to 2019, he was the Chair for Communication Theory with TU Dresden. He is currently the Managing Director of the Institute of Communications Technology, the Head of the Chair for Communications Systems, and a Full Professor with Technische Universit\"at Braunschweig, Brunswick, Germany. He has published more than 180 journal articles, 15 book chapters, one book, three monographs, and some 300 conference papers. His main research interests are in the broad area of communications. He was a recipient of the IEEE Signal Processing Society Best Paper Award. He and his colleagues were also recipients of the Best Paper and Best Student Paper Awards at the IEEE CAMSAP 2011, IEEE WCSP 2012, IEEE SPAWC 2012, IEEE ICUFN 2018, PETS 2019, and ISWCS 2019. Since 2017, he has been the Editor-in-Chief of the EURASIP Journal on Wireless Communications and Networking. He is currently serving on the editorial boards of the IEEE TRANSACTIONS ON INFORMATION THEORY and IEEE TRANSACTIONS ON COMMUNICATIONS. He was on the editorial boards of the IEEE SIGNAL PROCESSING LETTERS, the IEEE TRANSACTIONS ON SIGNAL PROCESSING, the IEEE TRANSACTIONS ON WIRELESS COMMUNICATIONS, and the IEEE TRANSACTIONS ON INFORMATION FORENSICS AND SECURITY.
\end{IEEEbiography}

\begin{IEEEbiography}[{\includegraphics[width=1in,height=1.25in,clip,keepaspectratio]{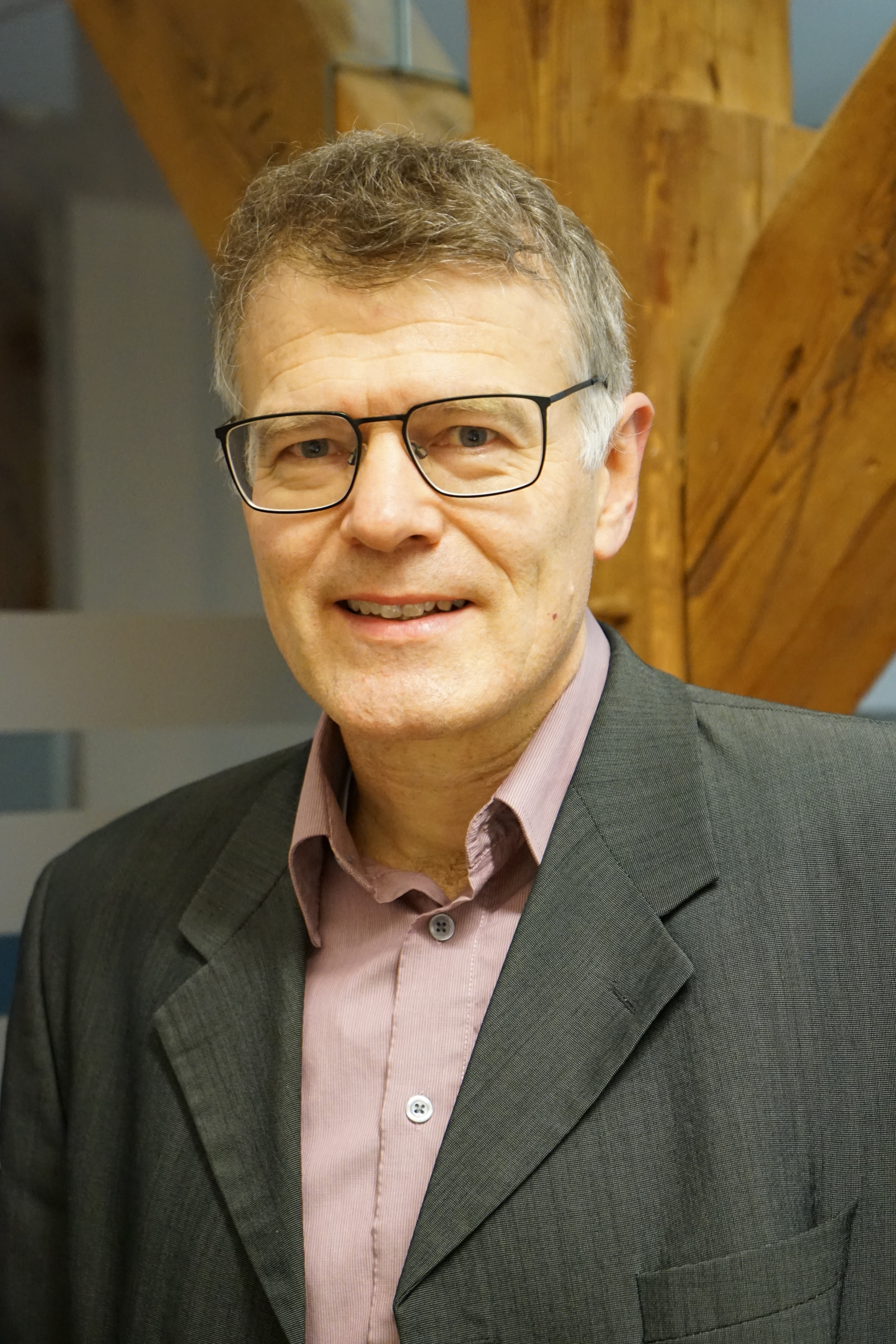}}]
{Robert Schober} (S'98, M'01, SM'08, F'10) received the Diplom (Univ.) and the Ph.D. degrees in electrical engineering from Friedrich-Alexander University of Erlangen-Nuremberg (FAU), Germany, in 1997 and 2000, respectively. From 2002 to 2011, he was a Professor and Canada Research Chair at the University of British Columbia (UBC), Vancouver, Canada. Since January 2012 he is an Alexander von Humboldt Professor and the Chair for Digital Communication at FAU. His research interests fall into the broad areas of Communication Theory, Wireless and Molecular Communications, and Statistical Signal Processing.

Robert received several awards for his work including the 2002 Heinz Maier Leibnitz Award of the German Science Foundation (DFG), the 2004 Innovations Award of the Vodafone Foundation for Research in Mobile Communications, a 2006 UBC Killam Research Prize, a 2007 Wilhelm Friedrich Bessel Research Award of the Alexander von Humboldt Foundation, the 2008 Charles McDowell Award for Excellence in Research from UBC, a 2011 Alexander von Humboldt Professorship, a 2012 NSERC E.W.R. Stacie Fellowship, a 2017 Wireless Communications Recognition Award by the IEEE Wireless Communications Technical Committee, and the 2022 IEEE Vehicular Technology Society Stuart F. Meyer Memorial Award. Furthermore, he received numerous Best Paper Awards for his work including the 2022 ComSoc Stephen O. Rice Prize and the 2023 ComSoc Leonard G. Abraham Prize. Since 2017, he has been listed as a Highly Cited Researcher by the Web of Science. Robert is a Fellow of the Canadian Academy of Engineering, a Fellow of the Engineering Institute of Canada, and a Member of the German National Academy of Science and Engineering.

He served as Editor-in-Chief of the IEEE Transactions on Communications, VP Publications of the IEEE Communication Society (ComSoc), ComSoc Member at Large, and ComSoc Treasurer. Currently, he serves as Senior Editor of the Proceedings of the IEEE and as ComSoc President.

\end{IEEEbiography} 

\begin{IEEEbiography}[{\includegraphics[width=1in,height=1.25in,clip,keepaspectratio]{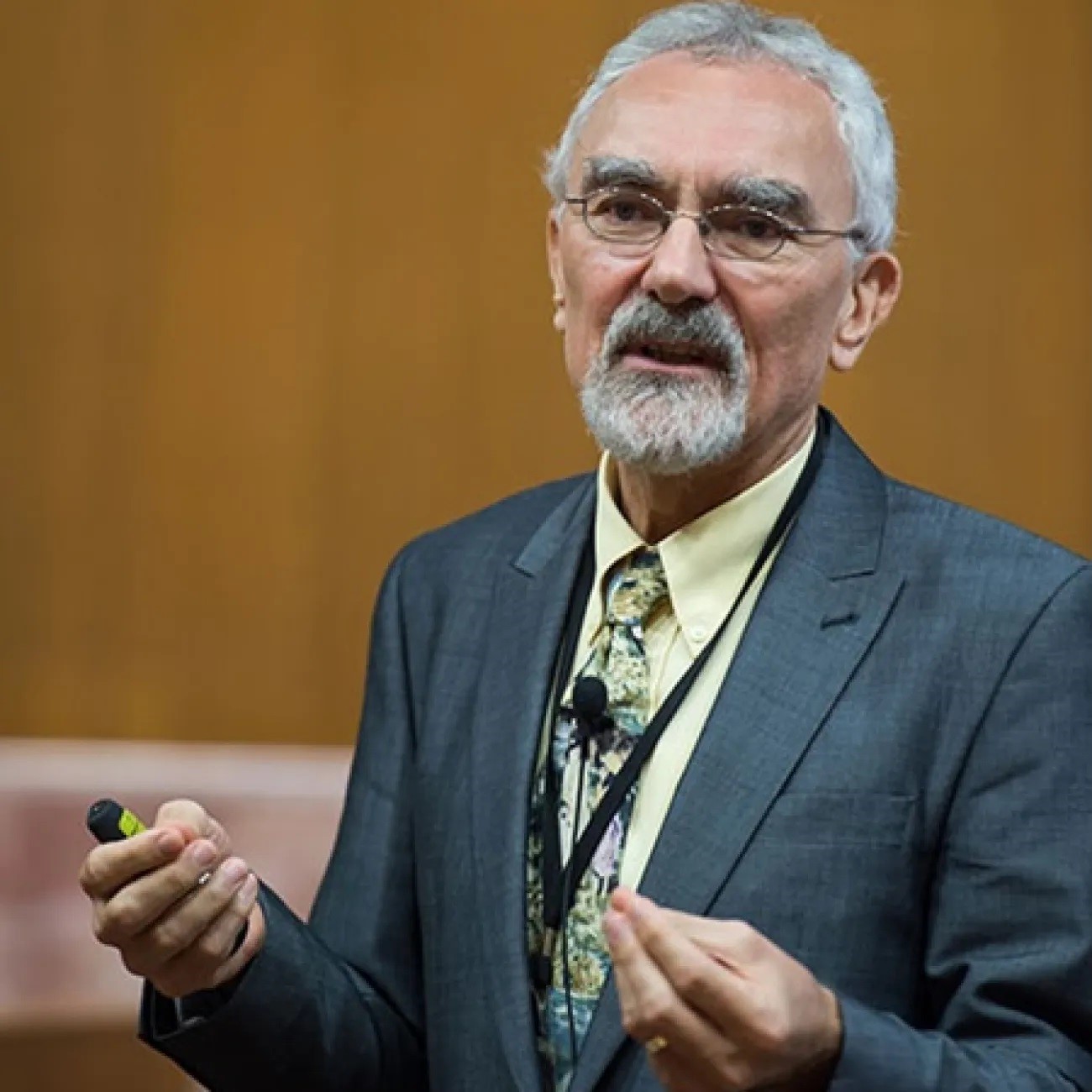}}]
{Lajos Hanzo}  is a Fellow of the Royal Academy of Engineering, FIEEE, FIET, Fellow of EURASIP and a Foreign Member of the Hungarian Academy of Sciences. He coauthored 2000+ contributions at IEEE Xplore and 19 Wiley-IEEE Press monographs. He was bestowed upon the IEEE Eric Sumner Technical Field Award.

\end{IEEEbiography}

\end{document}